\documentclass[a4paper,12pt]{article}
\usepackage{listings}
\usepackage{anysize}
\marginsize{2cm}{2cm}{2cm}{2cm}
\usepackage{amsmath, amsthm, amssymb, amsfonts} 
\usepackage{mathtools}
\usepackage{graphicx}
\usepackage{lscape}
\usepackage{tikz}
\usepackage{caption,subcaption}
\usetikzlibrary{shapes,arrows}
\usepackage{float}
\usetikzlibrary{decorations.pathmorphing} 
\usetikzlibrary{fit}					
\usetikzlibrary{backgrounds}
\usepackage{setspace}
\restylefloat{table}
\usepackage{epstopdf}
\usepackage{xcolor}
\usepackage{amsmath}
\usepackage{amssymb}
\usepackage{algorithm}
\usepackage{algpseudocode}

\newtheorem{theorem}{Theorem}

\newtheorem{corollary}[theorem]{Corollary}

\newtheorem{lemma}[theorem]{Lemma}

\newtheorem{proposition}[theorem]{Proposition}

\newtheorem{assumption}[theorem]{Assumption}

\makeatother
\usepackage{comment}
\usepackage{chngcntr}
\counterwithin{figure}{section}
\counterwithin{table}{section}
\counterwithin{theorem}{section}
\newcommand{\commentoD}[1]{\textcolor{orange}{[#1]}}
\newcommand{\commentoY}[1]{\textcolor{blue}{[#1]}}

\title{Pricing and Hedging Financial Derivatives in Merger\&Acquisition Deals with Price Impact}
\author{Emilio Barucci\thanks{Department of Mathematics, Politecnico di Milano}\and Yuheng Lan\thanks{Department of Mathematics, Politecnico di Milano} \and Daniele Marazzina\thanks{Department of Mathematics, Politecnico di Milano}
}
\begin{document}
\maketitle
\begin{abstract}
We investigate the optimal execution of contracts that are used in merger\&acquisition deals. We consider cash-settled and physically delivered contracts between a broker and a counterpart. Contracts are linear (total returns swaps), nonlinear (collar contracts) or Asian type (TWAP based contracts). We derive the optimal execution strategy and the optimal fee through indifference utility arguments allowing for linear market effects of trades. We show that linear cash-settled contracts are more expensive and more exposed to manipulation/statistical arbitrages by the broker. Also nonlinear and Asian type contracts are exposed to these phenomena. 
    \end{abstract}

\newpage
\section{Introduction}
Total Return Swaps (TRS) are financial contracts widely used in merger\&acquisition deals. 
Instead of directly acquiring stocks on the market bearing the market impact loss, the acquiring company enters a financial derivative with a broker asking for the delivery of the stocks of the target company at a certain date. 

As a pure illustrative example, we may refer to the Unicredit-Commerzbank experience, see \cite{STOR}. 
In 2024, Unicredit entered into TRS with Barclays and Bank of America, covering 5\% and 6.53\% stakes of Commerzbank. The goal was to pool a large share of stocks while meeting Eurozone regulatory requirements that prohibit acquiring more than 10\% of a bank without the approval by the European Central Bank. Additionally, through call and put options, Unicredit sets price limits on the shares covered by these contracts (collar-type contracts). This strategy allowed Unicredit to increase its direct and indirect ownership in Commerzbank from 9.2\% to 20.73\% conditionally on approval from the regulatory authority, making it (potentially) the bank's largest shareholder surpassing even the German government.

These agreements present several advantages for both parties. The broker is allowed to tame the trades at its best convenience.
Also the acquiring company (counterpart in what follows) gets flexibility to cope with regulatory approval (banking or competition authorities) or change of conditions. 

In this paper, we investigate the pricing and hedging of these contracts by a broker in a model with temporary and permanent market impact. The broker at the same time sets a fee for the contract through utility indifference pricing and trades the underlying asset hedging the position; we consider both cash-settled and physically delivered contracts. 
Physical delivery means that the broker has to deliver the stocks at a certain date, cash settlement means that at maturity there is only the exchange of cash flows and that the position is fully unwind at that time. 

The paper contributes to several strands of literature. 

First of all, we refer to the literature on pricing and hedging illiquid European/American derivatives, see \cite{BA-BA,BANK,LI-YO,KR-KU,BOUCH,BOUCH2,GUE17,ROCH,NAU,ROG,EKR,BERCH}. Most of these papers deal with cash-settled contracts, only some of them consider physical delivery.
These papers concentrate either on super replication/quadratic hedging of derivatives or on pricing through indifference utility. 
In our analysis, we extend the indifference utility approach of \cite{GUE17} considering linear, collar type and TWAP-based (Time-Weighted Average Price) contracts that are widely used in this type of deals, see \cite{BALD,BAN,BIC,LARS}. We confirm their result that cash-settled contracts are more expensive than physical delivery ones, the rationale being that in the first case the broker has to liquidate the underlying position at the terminal date.

We also contribute to the literature on market manipulation and statistical arbitrage in market models with price impact. \cite{HUBER,GATH0,ASS,GATH,MUL} have shown that market impact models with martingale dynamics allow for price manipulation and statistical arbitrage strategies, i.e., round trip trades with strictly positive expected payoffs. These papers provide conditions on the temporary and permanent price effect of a trade avoiding manipulation; instead, our analysis deals with manipulation in a market with permanent/temporary price impact by a trader hedging a long position on a TRS/collar-type contract. 
The literature on manipulation by a trader holding a long/short position on a derivative includes \cite{JARR,HO-NA,NY-PA,ROCH}. In our analysis,  manipulation concerns the sign of trades (the broker doesn't trade consistently on the same side of the market) and arbitrage deals with the possibility of implementing a strategy with a positive expected payoff. 

Finally, the paper is related to the literature on the manipulation of the underlying asset by a trader holding/trading a future, a contract which is quite similar to a TRS. 
Manipulation can be based on information or not, see \cite{KY-VI}; in the first case the trader exploits the presence of private information in the market to build a profitable manipulation strategy; in the latter the trader exploits market illiquidity to implement a wide spectrum of profitable strategies such as punch the close/settlement price with a positive payoff at maturity, pump and dump strategies, squeezes or corners. It is argued that cash settlement should be less subject to manipulation than physical delivery, see \cite{KYLE2,PIRRO}, but there are no definitive results on this issue, there are papers showing that manipulation may occur both with cash settlement and physical delivery, see \cite{KUMAR,HO-NA}. We show that cash-settled contracts are more exposed to manipulation because the broker at the same time hedges the stock exposure at the beginning of the trading period and then unwinds the position at the settlement date , while a physical delivery agreement aligns hedging and holdings. A nonlinear payoff, such as a collar contract, may render a much more complex trading strategy dependent on the underlying asset evolution. The rationale of these results is that the broker's activity is motivated by two main motivations: hedging the stock position and handling the inventory to meet the delivery of the stocks. The two motivations can conflict each other in case of cash settlement: the first one dominates at the beginning of the trading horizon, and the second one at the end. Also a TWAP-based contract may induce complex trading strategies. A linear and cash-settled contracts allow for a statistical arbitrage, instead physically settled contracts do not. TWAP-based contracts are associated to statistical arbitrages. 

Summing up, we provide two main results: the cash settlement option is more expensive than the physically delivered one; manipulation and statistical arbitrages are more likely to emerge in case of cash-settled contracts rather than physically delivered contracts.  

The paper is organized as follows. In Section \ref{MOD} we present the model and the contract agreements that are used in  merger\&acquisition deals. In Section \ref{PROB} we address the mathematical problem.
In Section \ref{NUM} we provide the numerical analysis.
In Section \ref{EXT} we analyze two extensions of the model (regulatory approval, TWAP based contracts).
In Appendix \ref{CLOSE} we provide closed form solutions for linear contracts in case of null drift for the stock price evolution dynamics and risk free rate. In Appendix \ref{Numerical method} we detail the numerical solution technique. In Appendix \ref{additional} we provide some additional Figures and Tables.

\section{The Model}
\label{MOD}
We consider contractual agreements that foresee the delivery at time $T$ by a broker of cash or stocks to the counterpart. At the same time, the broker hedges the position, by trading on the stock market, and determines the fee for the contract.

We consider a filtered probability space $(\Omega,\mathcal{F}, \{\mathcal{F}_t\}_{t\geq 0}, \mathbb{P})$, where $\{\mathcal{F}_t\}_{t\geq 0}$ denotes the filtration that represents the information available up to time $t$.

The broker's inventory of the stock evolves as follows:
\begin{equation}
    Q(t)=Q(0)+\int_0^{t}v(s)ds,
\end{equation}
where the trading speed $v$ belongs to the set of admissible strategies $\mathcal{A}$ defined as
\begin{equation*}
    \mathcal{A} := \left\{ v \text{ is } \mathcal{F}\text{-progressively measurable and } |v(t)| \leq C \text{ on } [0,T] \times \Omega \right\}.
\end{equation*}

We consider the simplest linear market impact model, see \cite{ALGCH}.
The stock price is influenced in a permanent way by the trades performed by the broker:
\begin{equation}
    dS(t)=(\mu+bv(t)) dt+\sigma dW(t),
\end{equation}
where $b$ is permanent market impact, $\mu$ and $\sigma$ are the drift and the diffusion coefficients. \\
Assuming that the broker has access to the risk free investment at a rate $r$, the wealth evolves as follows: 
\begin{equation}
\label{WEALTH}
    dX(t)=rX(t)dt-v(t)(S(t)+lv(t))dt,
\end{equation}
where $l$ is the temporary market impact. 

We also define a function $L(a, b)$ to represent the cost of liquidating the inventory of the stock from level $a$ to $b$, incorporating frictions, transaction fees and other market effects. Following \cite{CAR15}, we set 
\begin{equation}\label{eq_L}    
L(a,b)=\alpha(a-b)^2,
\end{equation}
where $\alpha>0$ is the penalty coefficient.

We consider several different contractual agreements. In what follows, we define the contractual terms and the payoffs for the broker. 
For each contract, the broker receives a fee at $t=0$ to be determined through utility indifference arguments. Once the contract is signed, the broker hedges the stock position and receives a payoff at time $T$ associated with the contractual agreement. In all the payoffs we insert the wealth of the broker at maturity $X(T)$ that includes the wealth associated with trading.
Through backward in time arguments, the fee is determined taking into account the hedging strategy.

We consider both physical delivery and cash settlement agreements for a certain amount of stocks ($N$) at time $T$. In case of physical delivery, the broker is obliged to deliver the stocks at time $T$ acquiring what is left over on the market ($N-Q(T)$), or selling any surplus if $Q(T)>N$. In case of cash settlement we assume that at time $T$ the broker unwinds the stock position $Q(T)$. In all the cases, we assume that the broker ends with no stocks.

\paragraph{1. Physical delivery}
The broker delivers $N$ stocks at time $T$ to the counterpart. Given the amount of stocks $Q(T)$ at time $T$, the broker must acquire (sell if negative) $N-Q(T)$ stocks on the market with a disbursement of $(N-Q(T))S(T)$ and a liquidity cost $L(Q(T),N)$. As a result, the payoff for  the broker is:
\begin{equation}\label{payoff_physical}
Y^{Physical}(T)=X(T)-(N-Q(T))S(T)-L(Q(T),N).
\end{equation}
The counterpart gets $N$ stocks with book value $NS(T)$.

\paragraph{2. Total Return Swap}
The counterpart enters a TRS on the stock with the broker. The TRS contract references $N$ shares of the underlying stock, with an initial price $S(0)$.
The TRS foresees the exchange of the appreciation of the stock (from the broker to the counterpart), i.e., the broker pays $NS(T)$ and receives $NS(0)$ (we ignore the risk free rate appreciation).
At time $T$, the broker unwinds the full position on the underlying stock with a reward $Q(T)S(T)$ and a liquidation cost $L(Q(T),0)$. Therefore, the payoff for the broker turns out to be: 
\begin{equation}
Y^{TRS}(T)=X(T)-N(S(T)-S(0))+Q(T)S(T)-L(Q(T),0).
\end{equation}
The TRS is a cash-settled agreement. It can be interpreted as the cash settlement version of the physical delivery delivery of $N$ stocks analyzed above. The only difference is that the broker receives $NS(0)$. In what follows, to compare the two contracts, we eliminate this component in the payoff (which is going to be included in the fee to be defined through indifference utility), and we consider the following:
\begin{equation}
\label{cash settlement}
Y^{Cash}(T)=X(T)-(N-Q(T))S(T)-L(Q(T),0).
\end{equation}
Notice that the only difference between the two contracts is provided by the liquidation costs. The counterpart gets $NS(T)$ in cash.

\paragraph{3. Collar (cash settlement)}
The counterpart may want to limit the exposure to market movements. In this case, the cash flow between the broker and the counterpart is a nonlinear function of the asset price at maturity. The most popular way to achieve this goal is through a collar agreement. At time $T$ the payoff is given by three components: long position on the underlying stock and on an European put option with strike price $K_1$, short position on an European call option with strike price $K_2$, where $K_1 < K_2$:
$$
Z(T)=S(T)+[K_1-S(T)]^+-[S(T)-K_2]^+,
$$
being $[\cdot]^+=\max\{\cdot,0\}$. The payoff structure limits the potential upside on the stock (up to $K_2$) in exchange for downside protection (below $K_1$). Both the put and call options share the same maturity date $T$. 

The agreement is purely financial (cash settlement). Instead of receiving the appreciation or depreciation of the stock (as in the TRS), the payoff of the counterpart is $NZ(T)$ in cash and the broker's cash flow turns out to be 
\begin{equation}\label{payoff_Collar}
Y^{Collar (cash)}(T)=X(T)-NZ(T)+Q(T)S(T)-L(Q(T),0).
\end{equation}

\paragraph{4. Collar (physical delivery)}
The broker signs a physical delivery contract for $N$ stocks, sells a collar and receives the market price for $N$ stocks, in this way the 
broker is short of a collar as in the cash settlement case:
\begin{equation}\label{payoff_physical_collar}
Y^{Collar (physical)}(T)=X(T)-(N-Q(T))S(T)-L(Q(T),N)+N(S(T)-Z(T))
\end{equation}
$$
=X(T)-NZ(T)+Q(T)S(T)-L(Q(T),N).
$$
The payoff of the counterpart is provided by $N$ stocks with book value $NS(T)$ and $N(Z(T)-S(T))$ in cash, as a result (summing book and market values) it is $NZ(T)$.


\section{Solution of the problem}
\label{PROB}
The broker addresses two problems: the optimal hedging and the definition of the fee for the contract through utility indifference pricing. 


We define the value function $\mathcal{J}$ as follows
\begin{equation*}
    \mathcal{J}(t,x,q,S)=\sup_{v\in\mathcal{A}}\mathbb{E}\left[-\exp(-\gamma Y(T))| X(t)=x, Q(t)=q, S(t)=S\right],\ t\leq T,
\end{equation*}
where $\gamma$ is the risk aversion coefficient, $v$ is a strategy and $\mathcal{A}$ is the set of the admissible strategies. 

For each contract, the broker's fee is determined through utility indifference. Following \cite[Section 3]{GUE17}, we assume that at $t=0$:
\begin{itemize}
    \item the broker enters into the contract with initial wealth $X_0$ and zero inventory of stocks, the contract is signed for a fee $\mathcal{P}$;
    \item the counterpart gives $q_0$ stocks to the broker, pays him the fee $\mathcal{P}$ and receives $q_0S_0$ in cash from the trader.
\end{itemize}

Following \cite[Section~2.3]{CAR} and \cite[Section~3]{GUE17}, the indifference fee $\mathcal{P}$ is defined as the amount of money that makes the broker with initial wealth $X_0$ equally satisfied between two alternatives:  
\begin{enumerate}
    \item to invest $X_0$ at the risk-free rate, which yields the terminal utility
\begin{equation}\label{ind_cond}  
        -\exp(-\gamma e^{rT} X_0),
    \end{equation}
    \item to enter the contract receiving  $\mathcal{P} - q_0 S_0$ in cash and $q_0$ shares, in exchange of  delivering the shares and/or the cash payoff at time $T$.
\end{enumerate}

In our setting, direct computation yields
\begin{align*}
    \mathcal{J}(0, X_0 + \mathcal{P}- q_0 S_0 , q_0, S_0)
    &= -\exp(-\gamma e^{rT}(X_0 + \mathcal{P})) \, \theta(0, q_0, S_0) \\
    &= -\exp\!\left[-\gamma e^{rT}\!\left(X_0 + \mathcal{P} - \frac{e^{-rT}}{\gamma}\log \theta(0, q_0, S_0)\right)\right].
\end{align*}
The existence of such a function $\theta$ is established in Section~\ref{section-viscosity solution}. Setting
\[
    \mathcal{P} = \frac{e^{-rT}}{\gamma}\log \theta(0, q_0, S_0),
\]
the above value function equates to $-\exp(-\gamma e^{rT} X_0)$. The argument can be extended to any  $t \in [0,T)$ to conclude that
\begin{equation}
    \mathcal{P}(t,q,S)
    = \frac{e^{-r(T-t)}}{\gamma}\log \theta(t,q,S)
    \label{indifference_price}
\end{equation}
is the utility indifference fee of the contract. 

The value functions for the four contracts are as follows: 
\begin{itemize}
    \item Physical delivery:
    \begin{equation*}
        \mathcal{J}_1(t,x,q,S)=\sup_{v\in\mathcal{A}}\mathbb{E}[-\exp(-\gamma(X(T)-(N-Q(T))S(T)-L(Q(T), N))],
    \end{equation*}
    \item TRS:
    \begin{equation*}
        \mathcal{J}_2(t,x,q,S)=\sup_{v\in\mathcal{A}}\mathbb{E}[-\exp(-\gamma(X(T)-(N-Q(T))S(T)-L(Q(T), 0))].
    \end{equation*}
    \item Collar (cash):
    \begin{align*}
\mathcal{J}_3(t,x,q,S)=&\sup_{v\in\mathcal{A}}\mathbb{E}[-\exp(-\gamma(X(T)-(N-Q(T))S(T)\\&+N[S(T)-K_2]^{+}-N[K_1-S(T)]^{+}-L(Q(T),0)].
    \end{align*}
    \item Collar (physical):
    \begin{align*}
\mathcal{J}_4(t,x,q,S)=&\sup_{v\in\mathcal{A}}\mathbb{E}[-\exp(-\gamma(X(T)-(N-Q(T))S(T)\\&+N[S(T)-K_2]^{+}-N[K_1-S(T)]^{+}-L(Q(T),N)].
    \end{align*}
\end{itemize}

\subsection{Characterization of the value function}\label{section-viscosity solution}
The above payoffs can be written in a common framework as follows:
\begin{equation}\label{value function}
    \mathcal{J}(t,x,q,S)=\sup_{v\in\mathcal{A}}\mathbb{E}[-\exp(-\gamma(X(T)+Q(T)S(T)-\Pi(S(T))-L(Q(T),\cdot))]
\end{equation}
where $L(q,\cdot):\mathbb{R}\times \{0,N\} \rightarrow\mathbb{R}$ represents the market friction costs at maturity, depending on the cash or physical settlement feature, and $\Pi(S):\mathbb{R}\rightarrow\mathbb{R}$ is the payoff associated with the financial contract. 

We make the following assumption on $L(q)$ (with an abuse of notation we avoid the second argument) and $\Pi(S)$.

\begin{assumption}\label{assumption on function}
    $L(q)$ is defined as in Equation (\ref{eq_L}), $\Pi(S)$ is a globally Lipschitz-continuous function.
\end{assumption}
All the four contracts satisfy this assumption. In the following, we use the viscosity solution framework to characterize the value function. 

For $t\in[0,T]$ and $x,q,S,p_x,p_q,p_S,M_{SS}\in\mathbb{R}$, the Hamiltonian of the stochastic control problem is
\begin{equation}\label{Hamiltonian}
    \mathcal{H}(t,x,q,S,p_x,p_q,p_S,M_{SS})=\sup_{|v|\leq C}\{(\mu+bv)p_{S}+(rx-(S+lv)v)p_{x}+vp_q+\frac{1}{2}\sigma^2M_{SS}\}.
\end{equation}

As
    \begin{equation*}
        d(e^{-rt}X(t))=e^{-rt}(-v(t)(S(t)+lv(t))dt),
    \end{equation*}
we can compute
\begin{align*}
    e^{-rT}X(T)-e^{-rt}X(t)
    &=\int_t^{T}e^{-ru}(-v(u)(S(u)+lv(u))\,du)\\
    &= -\int_{t}^{T}e^{-ru}S(u)\,dQ(u)-\int_{t}^{T}e^{-ru}lv^2(u)\,du\\
    &= -e^{-ru}S(u)Q(u)\bigg|^{T}_t + \int_{t}^{T}e^{-ru}Q(u)\,dS(u)
    -\int_{t}^{T} r e^{-ru}S(u)Q(u)\,du \\
    &\quad -\int_{t}^{T}e^{-ru}lv^2(u)\,du\\
    &=e^{-rt}Q(t)S(t)-e^{-rT}Q(T)S(T)\\
    &\quad +\int_t^{T}e^{-ru}\big[(-rS(u)+\mu+bv(u))Q(u)-lv^2(u)\big]\,du
    +\int_t^{T}\sigma e^{-ru}Q(u)\,dW(u)
\end{align*}
    and therefore, since $X(t)=x,\,Q(t)=q,$ and, with an abuse of notation, $S(t)=S$, we get
    \begin{align*}
        X(T)+Q(T)S(T)=e^{r(T-t)}(x+qS)+e^{r(T-t)}\left(\int_t^{T}e^{-r(u-t)}\left[(-rS(u)+\mu+bv(u))Q(u)-lv^2(u)\right]du\right.\\
        \left.+\int_t^{T}\sigma e^{-r(u-t)}Q(u)dW(u)\right).
    \end{align*}
 Then the value function in Equation (\ref{value function}) can be written as:
    \begin{equation}\label{transformation}
        \mathcal{J}(t,x,q,S)=-\exp(-\gamma e^{r(T-t)}(x+qS))\theta(t,q,S)
    \end{equation}
    where
    \begin{align*}
    \theta(t,q,S)=\inf_{v\in\mathcal{A}}\mathbb{E}\left[ \exp\left(-\gamma e^{r(T-t)}\left\{\int_t^{T}e^{-r(u-t)}\left[(-rS(u)+\mu+bv(u))Q(u)-lv^2(u)\right]du\right.\right.\right.\\
     \left.\left.  \left. +\int_t^{T}\sigma e^{-r(u-t)}Q(u)dW(u) -e^{-r(T-t)}\Pi(S(T))-e^{-r(T-t)}L(Q(T))\right\}\right)\right].
    \end{align*}
    To simplify the notation, we define the random variable $\Gamma(t,q,S,v)$ as
    \begin{align*}
    \Gamma(t,q,S,v)=
\left\{\int_t^{T}e^{-r(u-t)}\left[(-rS(u)+\mu+bv(u))Q(u)-lv^2(u)\right]du\right.\\
     \left.
      +\int_t^{T}\sigma e^{-r(u-t)}Q(u)dW(u) -e^{-r(T-t)}\Pi(S(T))-e^{-r(T-t)}L(Q(T))\right\}
        \end{align*}
    and therefore
    \begin{equation}\label{theta}
\theta(t,q,S)=\inf_{v\in\mathcal{A}}\mathbb{E}\left[\exp(-\gamma e^{r(T-t)}\Gamma(t,q,S,v))\right].
    \end{equation}

The following regularity result can be established.

\begin{lemma}\label{locally bounded}
For any $(t,q,S)\in[0,T]\times\mathbb{R}\times\mathbb{R}$, the function $\theta(t,q,S)$ is locally bounded.
\end{lemma}

\begin{proof}
Since $\theta(t,q,S)\ge 0$ by definition, it suffices to prove a local upper bound. Let $D \subset [0,T]\times\mathbb{R}\times\mathbb{R}$ be compact.
By definition,
\[
\theta(t,q,S)=\inf_{v\in\mathcal{A}} \mathbb{E}\Big[\exp\big(-\gamma e^{r(T-t)}\Gamma(t,q,S,v)\big)\Big].
\]
Choosing $v\equiv 0$, we obtain
\[
\theta(t,q,S)\le 
\mathbb{E}\Big[\exp\big(-\gamma e^{r(T-t)}\Gamma(t,q,S,0)\big)\Big].
\]
Under $v\equiv 0$, the inventory is constant and $S$ is an arithmetic Brownian motion. Using that $\Pi(S)$ is globally Lipschitz and $L(q)$ has quadratic growth, that is, there exists $c_1>0$ such that
\[
L(q)\le c_1(1+|q|^2),
\]
a direct estimate yields
\[
|\Gamma(t,q,S,0)|
\le
c_2\Big(1+\sup_{u\in[t,T]}|S(u)|\Big),
\qquad (t,q,S)\in D,
\]
for some constant $c_2>0$.
Therefore, there exists $c_3>0$ such that
\[
\exp\big(-\gamma e^{r(T-t)}\Gamma(t,q,S,0)\big)
\le
\exp\Big(c_3\big(1+\sup_{u\in[t,T]}|S(u)|\big)\Big).
\]
Since the running supremum of an arithmetic Brownian motion over a finite time interval has finite exponential moments, it follows that 
\[
\sup_{(t,q,S)\in D}
\mathbb{E}\Big[\exp\big(c_3\sup_{u\in[t,T]}|S(u)|\big)\Big] < \infty.
\]
Hence, we proved that there exists $c_D>0$ such that
\[
\theta(t,q,S)\le c_D,
\qquad (t,q,S)\in D.
\]
\end{proof}

\begin{proposition}[Viscosity solution]\label{viscosity solution}
Under Assumption \ref{assumption on function}, the value function $\mathcal{J}$ is the unique viscosity solution of the HJB equation
\begin{equation}\label{HJB}
    \partial_t\mathcal{J}+\mu\partial_S\mathcal{J}+\frac{1}{2}\sigma^2\partial_{SS}\mathcal{J}
    +\sup_{|v|\le C}\Big\{bv\partial_S\mathcal{J}+(rx-(S+lv)v)\partial_x\mathcal{J}+v\partial_q\mathcal{J}\Big\}=0
\end{equation}
on $[0,T)\times\mathbb{R}\times\mathbb{R}\times\mathbb{R}$, with terminal condition
\[
    \mathcal{J}(T,x,q,S)=-\exp\big(-\gamma(x+qS-\Pi(S)-L(q))\big).
\]
\end{proposition}
\begin{proof}
The proof follows the standard viscosity-solution approach for stochastic control problems; see \cite[Chapter 4]{PHA}.
First, by the representation formula
\[
\mathcal{J}(t,x,q,S)
=
-\exp\!\big(-\gamma e^{r(T-t)}(x+qS)\big)\theta(t,q,S),
\]
since the exponential prefactor is continuous and locally bounded, and $\theta$ is locally bounded by Lemma \ref{locally bounded}, it follows that $\mathcal{J}$ is locally bounded on $[0,T]\times\mathbb{R}\times\mathbb{R}\times\mathbb{R}$.
Moreover, the control set is the compact set $U=[-C,C]$, the state dynamics have continuous coefficients, and such coefficients satisfy the standard linear-growth condition. Therefore, the dynamic programming principle applies. By \cite[Proposition 4.3.1]{PHA}, the value function $\mathcal{J}$ is a viscosity supersolution of Equation \eqref{HJB}. By \cite[Proposition 4.3.2]{PHA}, $\mathcal{J}$ is also a viscosity subsolution of Equation \eqref{HJB}. Hence $\mathcal{J}$ is a viscosity solution of Equation \eqref{HJB} on $[0,T)\times\mathbb{R}^3$.\\
The terminal payoff is given by
\[
g(x,q,S):=-\exp\big(-\gamma(x+qS-\Pi(S)-L(q))\big).
\]
Under Assumption \ref{assumption on function}, the function $g$ is continuous. Moreover, the finiteness of the Hamiltonian follows from the compactness of the control set and the continuity of the coefficients with respect to $v$. Therefore, we can apply \cite[Theorem 4.3.2]{PHA}, which ensures that the value function satisfies the terminal condition in the viscosity sense. In particular, since $g$ is continuous, this implies that the terminal condition holds in the classical sense:
\[
\mathcal{J}(T,x,q,S)=g(x,q,S).
\]
Finally, uniqueness follows from the comparison principle for viscosity solutions, see \cite{DAL}. Therefore, $\mathcal{J}$ is the unique viscosity solution of Equation \eqref{HJB} with the above terminal condition.
\end{proof}

\subsection{Optimal trading strategy}

We now derive a formal characterization of the optimal control from the Hamiltonian of the HJB equation for $\mathcal J$. From Equation \eqref{HJB}, the control enters through the term
\[
\sup_{|v|\le C}
\Big\{
bv\,\partial_S\mathcal J-(S+lv)v\partial_x\mathcal J+v\partial_q\mathcal J
\Big\}.
\]
Rearranging the terms, this can be written as
\[
\sup_{|v|\le C}
\Big\{
-l\,\partial_x\mathcal J \, v^2
+
\big(b\partial_S\mathcal J - S\partial_x\mathcal J + \partial_q\mathcal J\big)v
\Big\}.
\]
Since $\partial_x\mathcal J=-\gamma e^{r(T-t)}\,\mathcal J(t,x,q,S)>0$, see Equation (\ref{transformation}), the above expression is strictly concave in $v$, and therefore admits a unique maximizer in $[-C,C]$.

The first-order condition for the unconstrained problem yields
\[
v^{\mathrm{unc}}(t,x,q,S)
=
\frac{
b\partial_S\mathcal J - S\partial_x\mathcal J + \partial_q\mathcal J
}{
2l\,\partial_x\mathcal J
}.
\]
Taking into account the constraint $|v|\le C$, the optimal control is given by
\begin{equation}\label{optimal_control}
v^*(t,x,q,S)
=\min\left\{C; \max\left\{-C;
\frac{
b\partial_S\mathcal J - S\partial_x\mathcal J + \partial_q\mathcal J
}{
2l\,\partial_x\mathcal J
}
\right\}\right\}.
\end{equation}

\subsection{Utility indifference fee}\label{section-utility indifference fee}
We are now able to compute the fees of the contracts through indifference utility arguments.

\begin{proposition}(Viscosity Solution of Utility indifference fee)
Under Assumption \ref{assumption on function}, $\forall(t,q,S)\in[0,T]\times\mathbb{R}\times\mathbb{R}$, $\mathcal{P}(t,q,S)$ is the viscosity solution of the associated HJB equation:
\begin{eqnarray}\label{indifferencepriceHJB}
    &&-\partial_t\mathcal{P}+r\mathcal{P}+(\mu-rS)q-\mu\partial_S\mathcal{P}-\frac{1}{2}\sigma^2\partial_{SS}\mathcal{P}-\frac{1}{2}\sigma^2\gamma e^{r(T-t)}(q-\partial_S\mathcal{P})^2\\&&\quad\quad\quad\quad+\sup_{|v|\leq C}\{-lv^2+(bq-b\partial_S\mathcal{P}-\partial_q\mathcal{P})v\}=0\nonumber
\end{eqnarray}
satisfying the terminal condition
\begin{equation*}
    \mathcal{P}(T,q,S)=\Pi(S)+L(q).
\end{equation*}
\end{proposition}
\begin{proof}
By Proposition \ref{viscosity solution}, $\mathcal{J}(t,x,q,S)$ is a viscosity solution of
\begin{equation*}
    \partial_t\mathcal{J}+\mu\partial_S\mathcal{J}+\frac{1}{2}\sigma^2\partial_{SS}\mathcal{J}+\sup_{|v|\leq C}\{bv\partial_S\mathcal{J}+(rx-(S+lv)v)\partial_x\mathcal{J}+v\partial_q\mathcal{J}\}=0.
\end{equation*}
Equations (\ref{indifference_price}) and (\ref{transformation}) imply that
\begin{equation}\label{fromJto_P}
\mathcal{J}(t,x,q,S)=-\exp[-\gamma e^{r(T-t)}(x+qS-\mathcal{P}(t,q,S))].
\end{equation}
By direct computation, we obtain
\begin{align*}
    \partial_t\mathcal{J}&=(-r(x+qS-\mathcal{P})-\partial_t\mathcal{P})\gamma e^{r(T-t)}\exp[-\gamma e^{r(T-t)}(x+qS-\mathcal{P})], \\
    \partial_S\mathcal{J}&=(q-\partial_S\mathcal{P})\gamma e^{r(T-t)}\exp[-\gamma e^{r(T-t)}(x+qS-\mathcal{P})], \\
    \partial_{SS}\mathcal{J}&=(-\gamma e^{r(T-t)}(q-\partial_S\mathcal{P})^2-\partial_{SS}\mathcal{P})\gamma e^{r(T-t)}\exp[-\gamma e^{r(T-t)}(x+qS-\mathcal{P})],\\
    \partial_x\mathcal{J}&=\gamma e^{r(T-t)}\exp[-\gamma e^{r(T-t)}(x+qS-\mathcal{P})], \\
    \partial_q\mathcal{J}&=(S-\partial_q\mathcal{P})\gamma e^{r(T-t)}\exp[-\gamma e^{r(T-t)}(x+qS-\mathcal{P})].
\end{align*}
Since the transformation \eqref{fromJto_P} is smooth and strictly monotone, the change of variables from $\mathcal{J}$ to $\mathcal{P}$ is valid in the viscosity sense. Therefore, the lower semicontinuous envelope and the upper semicontinuous envelope of $\mathcal{P}$ are, respectively, a viscosity supersolution and a viscosity subsolution of Equation \eqref{indifferencepriceHJB}. We conclude that $\mathcal{P}$ is a viscosity solution of  Equation \eqref{indifferencepriceHJB}.
\end{proof}

Notice that the optimal control provided in Equation (\ref{optimal_control}) can be rewritten as
\[\
v^*(t,q,S)
=\min\left\{C; \max\left\{-C;
\frac{
bq - b\partial_S\mathcal P(t,q,S) - \partial_q\mathcal P(t,q,S)
}{
2l
}
\right\}\right\}.
\].

In Appendix~\ref{CLOSE}, under the assumption $\mu = r = 0$, we derive closed-form expressions for the optimal control in the case of linear contracts in Equations ~\eqref{payoff_physical} and~\eqref{cash settlement}; see Propositions~\ref{CFPhys} and~\ref{CFTRS}. In Corollaries~\ref{PHYSICALprice} and~\ref{TRSprice} we also determine the corresponding optimal fees. In the subsequent numerical results section, when closed-form solutions are not available, we solve the problem numerically using a finite difference method; see Appendix~\ref{Numerical method} for further details.

\section{Numerical analysis}
\label{NUM}

In this section, we deal with the parameters reported in Table \ref{parame0}, where 
  $b, \ l, \ \alpha$ are as in \cite[Chapter 6.9]{CAR15}. The simulations are performed setting $S_0 = 45.0$ and $q_0 = 0.5$.
  
\begin{table}[tp]
	\centering
    \caption{Baseline set of parameters of the model.}
    \label{parame0}
	\begin{tabular}{|c|c|c|}
		\hline
		Parameter & Value & Interpretation \\
		\hline
        $r$ & 0.0& risk-free rate\\
        $\mu$ & 0.0& drift coefficients \\
		$b$ & $1e-3$ & permanent impact \\
		$l$ & $1e-3$ & temporary impact \\
		$\gamma$ & $1e-2$ & risk aversion \\
        $\sigma$ & 5 & volatility\\
		$N$ & 1 & quantity of shares \\
        $C$ & 10 & trading speed bound\\
        $\alpha$ & 0.2 & terminal penalty\\
        $T$ & 1 & terminal time\\
        $K_1$ & 40 & Lower execution price\\
        $K_2$ & 50 & Upper execution price\\
		\hline
	\end{tabular}\label{tab:params}
\end{table}
 First, in Figure~\ref{fig:simPhys_TRS}, we report the strategies in the interval $[0, T]$ for linear contracts (physical delivery and TRS). The optimal trading strategies do not depend on the evolution of the underlying asset price, as shown in Propositions \ref{CFPhys} and~\ref{CFTRS}. At the beginning of the trading period, the optimal solutions for physical delivery and TRS contracts look very similar.
In case of physical delivery, the inventory converges smoothly to the quantity of stocks that the broker has to deliver at maturity. In case of cash settlement, the need to unwind the position of the contract at maturity leads to an optimal strategy that first foresees to buy the stock and then to sell it ending with a null inventory at maturity. The path can be explained as an attempt by the broker to hedge the market risk nested in the cash-settled contract (at the beginning) and then to minimize the terminal penalty. 

In Appendix \ref{additional} we show the indifference fees and the optimal trading strategies for the two contracts at $t=0, \ 0.5, \ 0.95$ as a function of the inventory $q$ and asset price $S$. As expected Figures \ref{fig:phys}-\ref{fig:TRS} show that the optimal fee increases in the asset price and decreases in the inventory, the second effect being very small. Approaching maturity, the trading strategy becomes more aggressive. In the case of physical delivery, the agent increases the trading speed $v$ in order to bring the inventory closer to the target level $N$, potentially reaching the upper bound $C$ on the trading startegy when the inventory is low. Conversely, in the case of cash settlement, as maturity approaches the trading activity switches from positive to negative, and can reach the lower bound $-C$ when the inventory is close to $N$.

\begin{figure}[t]
    \centering
    \includegraphics[width=0.9\linewidth]{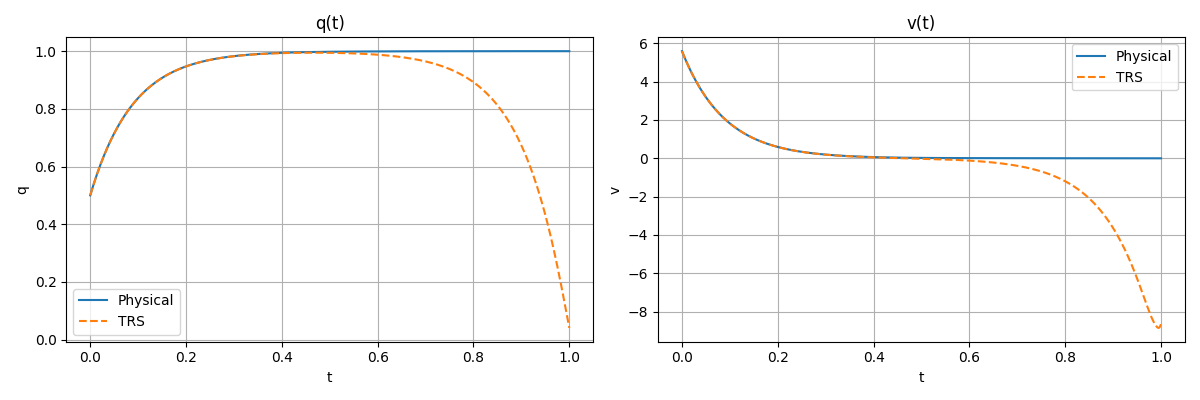}
    \caption{Inventory $q$ (left) and optimal strategy $v$ (right) for the physical delivery and TRS contracts. }
    \label{fig:simPhys_TRS}
\end{figure}

\begin{figure}[t]
    \centering
    \includegraphics[width=0.9\linewidth]{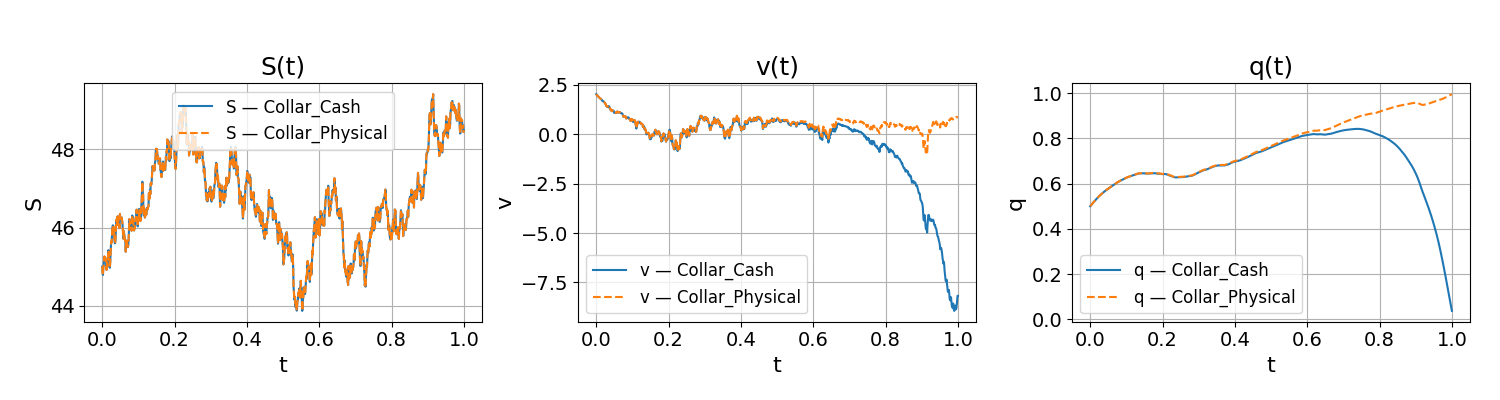}
    \caption{Simulations of the asset price $S$ (left), optimal strategy $v$ (middle) and inventory $q$ (right) for the collar contracts (physical delivery and cash settlement). The tests were conducted by simulating the same Brownian motion. 
    }\label{fig:simPhys_Collar}
\end{figure}

Figure~\ref{fig:simPhys_Collar} illustrates the results for the collar-type contract (both physical delivery and cash settlement) for a representative simulation of the asset price. In this setting, the trading strategy is sensitive to the evolution of $S$. The figure reports the dynamics of the stock price, the associated optimal trading strategy, and the resulting inventory process.
The trading strategies corresponding to collar contracts resemble those obtained for linear contracts, although they are generally less aggressive. As in the linear case, the cash-settled and physical delivery strategies are similar at the beginning of the trading horizon.
Figures~\ref{fig:col} and~\ref{fig:col2} in Appendix~\ref{additional} display the indifference fees and the optimal trading strategies as a function of the inventory and stock price at $t = 0,\, 0.5,\, 0.95$. At the beginning of the trading horizon ($t$ close to $0$), the broker tends to hold the underlying asset when the price lies within the interval $[K_1, K_2]$, i.e., when the contract is in the money and therefore the broker aims to hedge equity risk. Conversely, when the price approaches the barriers or moves outside this interval (i.e., the contract is out of the money), the broker tends to sell the underlying asset.
When the collar is out of the money, the three-dimensional surface representing the optimal trading strategy becomes smoother and exhibits an overall S-shaped profile. These figures highlight the importance of the hedging component at the beginning of the trading horizon: regardless of whether the contract involves physical or cash settlement, the broker holds the underlying asset when the collar is in the money. Moreover, the fee is S-shaped as a function of the stock price, and the corresponding surface flattens when the contract is out of the money (i.e., for $S < K_1 = 40$ or $S > K_2 = 50$).

In contrast with the findings of \cite{KYLE2,PIRRO}, Figure~\ref{fig:simPhys_TRS} shows that there is no evidence of manipulation (i.e., alternating buy and sell trades) in the case of a linear contract with physical delivery. By contrast, in the case of a cash-settled contract, the broker initially buys and subsequently sells the stock: the broker first hedges the equity risk and then liquidates the position as maturity approaches. This is due to the penalizing cost at maturity on stocks and is consistent with \cite{HO-NA} suggesting that physical settlement is less exposed to manipulation. As illustrated in Figure~\ref{fig:simPhys_Collar}, the nonlinearity of a collar-type contract leads to a more complex trading strategy, which depends on the evolution of the underlying asset price, see also Figures \ref{Q2} and \ref{Q3} in the Appendix.

\begin{table}[tp]
	\centering
    \caption{Fees of the four contracts.}
    \label{tab:IP}
	\begin{tabular}{|c|c|c|c|c|}
		\hline
		Model & physical delivery & TRS  & Collar (physical) & Collar (cash)\\
		\hline
		Price & $45.0029$ & $45.0130$  & $45.0042$ & $45.0078$ \\
		\hline
	\end{tabular}
\end{table}
\begin{table}[tp]
	\centering
    \caption{Expected payoffs of the four contracts.}
	\begin{tabular}{|c|c|}
		\hline
		Model &  $\mathbb{E}[Y(T)|X(0)=X_0-q_0S_0+\mathcal{P}(0,q_0,S_0)]-X_0$ \\
		\hline
          Physical delivery & -0.0137\\
          TRS &  0.0530\\
          Collar(physical) &  -0.0101\\
          Collar(cash) &  0.0527\\
        \hline
	\end{tabular}
    \label{STATARB}
\end{table}

In Table~\ref{tab:IP} we show the indifference fees of the four contracts. Confirming \cite{GUE17}, cash-settled contracts are more expensive than  physical delivery contracts.
The rationale of this result is that a cash-settled contract entails more trading activity and, therefore, higher trading costs.
 
 In Table \ref{STATARB} we compute the expected payoff assocaited with the contract:
 $$
 \mathbb{E}[Y(T)|X(0)=X_0-q_0S_0+\mathcal{P}(0,q_0,S_0)]-X_0,
 $$ 
 assuming $X_0=q_0S_0=22.5$, that is, after paying $q_0 S_0$ to acquire the shares, the broker's cash position at time $t=0$ is only due to the contract fee $\mathcal{P}(0,q_0,S_0)$. A positive expected payoff shows that, upon entering the contract and following the optimal strategy, the broker achieves a positive expected profit. 
 Such a result can be interpreted as statistical arbitrage opportunity, i.e., positive expected profit generated by the trading strategy associated with the contract, without additional capital beyond entering the contract. As shown in the table, a statistical arbitrage only arises for cash-settled contracts, and not for the physically delivered ones.

\subsection{Sensitivity analysis}

The analysis significantly depends on the model parameters. In what follows, we provide a sensitivity analysis on trading strategy, inventory and contract fee. For space considerations, all figures and tables are reported in Appendix~\ref{additionalSA}.

In Figure~\ref{fig:sweep_r}, we consider a positive interest rate, while the risky asset has zero drift. For all contracts, at the beginning of the trading horizon the broker sells the underlying asset, taking a short position in order to benefit from the positive risk-free rate. As maturity approaches, the broker gradually buys back the underlying asset to meet the terminal inventory constraint ($0$ in the case of cash settlement and to $1$ in the case of physical delivery).
Due to the positive interest rate, all contracts are cheaper compared to the case of a zero risk-free rate; see Table~\ref{fees}. The intuition is that the possibility of earning a positive return on the bank account reduces the compensation required by the broker.
Moreover, fees for cash-settled contracts become lower than those for physically delivered contracts. This effect is driven by the broker's trading behavior: initially, the broker sells the underlying asset and invests in the risk-free asset; as maturity approaches, the broker must repurchase the underlying asset to satisfy the terminal constraint, and the quantity to be acquired is significantly larger in the case of physical delivery. This results in higher trading costs for physically delivered contracts.

In Figure~\ref{fig:sweep_mu}, we consider a non-zero drift of the risky asset, while the risk-free rate is set to zero. When $\mu > 0$, the broker initially acquires the underlying asset in order to benefit from the positive price drift. The inventory reaches a level higher than $N$. As maturity approaches, the broker gradually sells the asset to meet the terminal inventory target (equal to $0$ or $1$). The opposite behavior is observed when $\mu < 0$. In this case, the broker initially takes a short equity position and later buys back the asset as maturity approaches. 
The presence of a non-zero drift reduces the fees of the contracts compared to the zero-drift case, as it provides additional gain from trading stocks.
The comparison between fees under physical delivery and cash settlement depends on the sign of $\mu$. When $\mu < 0$, the broker sells the asset at the beginning of the trading period and repurchases it close to maturity to meet the terminal constraint. Since this constraint is more stringent under physical delivery ($1$ instead of $0$), the corresponding fee is higher than in the cash-settled case. The opposite ordering is observed when $\mu > 0$.

In contrast with the benchmark case $\mu = r = 0$, we now observe that when either the drift or the risk-free rate is non-zero, manipulation (i.e., alternating buy and sell trades) may emerge also in the case of a linear contract with physical delivery, as the broker is incentivized to exploit favorable market conditions.

In Figure~\ref{fig:sweep_sigma}, we analyze the effect of changes in the volatility $\sigma$ of the underlying asset. The impact depends on the type of contract. 
For linear contracts (both cash-settled and physically delivered), higher volatility induces the broker to hold a larger inventory at the beginning of the trading horizon, in order to hedge market risk. In contrast, for collar contracts, an increase in $\sigma$ leads the broker to hold a smaller inventory. This difference arises because collar contracts impose both lower and upper bounds on the payoff, thereby reducing the broker’s exposure to market risk.
In Figure~\ref{fig: physical vs TRS under sigma}, we further investigate the optimal inventory for linear contracts as $\sigma$ increases. When volatility is high, the broker’s inventory is similar for physical delivery and cash-settled contracts at the beginning of the trading horizon, and then diverges as maturity approaches to meet the different terminal constraints. This reflects the fact that, under high volatility, hedging motives dominate initially, while liquidation requirements become relevant only closer to maturity.
By contrast, when volatility is low, the inventory dynamics differ from the outset. In particular, in the cash-settled case, the broker starts reducing the inventory from the beginning, as the hedging demand is limited.

In Figure \ref{fig:sweep_gamma} we consider the effect of varying the risk aversion coefficient. When $\gamma$ increases, the inventory increases, the rationale being that the broker takes care more carefully of market risk. 
As expected, as $\sigma$ or $\gamma$ increases, the fees of all contracts increase, see Table \ref{fees}. This can be explained by the fact that the broker faces a higher market risk and is more risk averse and, therefore, asks for a higher remuneration.

As the market impact cost $l$ increases, the broker trades at a slower rate and, in the case of cash-settled contracts, holds a lower inventory in order to reduce liquidation costs, see Figure~\ref{fig:sweep_l}. The permanent impact parameter $b$ plays a negligible effect, the corresponding results are therefore not reported.
As expected, increasing the terminal penalty parameter $\alpha$ does not affect the trading strategy at the beginning of the trading horizon. However, as maturity approaches, we observe a faster convergence towards the terminal inventory target ($0$ or $1$, depending on the settlement type) in order to avoid final liquidation costs, see Figure~\ref{fig:sweep_alpha}.
For all contracts, the indifference fee increases with the market impact parameter $l$, see Table~\ref{fees}. A similar effect is observed for the parameter $\alpha$, as a larger penalty increases the cost associated with failing to meet the terminal inventory target at maturity.

Notice that the absence of manipulation in the case of physical delivery contracts is robust to changes in volatility, risk aversion, impact costs, and the terminal penalty parameter.

\section{Extensions}
\label{EXT}
In this section we consider two extensions: the approval of the deal by the regulatory authority, see Section \ref{REG}, and agreements based on TWAP, see Section \ref{TWAP}.

\subsection{Regulatory approval}
\label{REG}
Let us consider a physical delivery contract with the possibility that at some time $\tau \in [0,T]$ the counterpart is forced to switch to cash settlement from a physical delivered agreement. We relate this event to regulatory approval that may induce the counterpart not to go on with a physical delivery contract. 






More precisely, the acquisition of the $N$ shares by the counterpart is subject to regulatory approval. The regulatory outcome is represented by a binary random variable $R\in\{0,1\}$, where $R=1$ corresponds to approval and $R=0$ to rejection. We assume that $R$ is $\mathcal F_\tau$-measurable, and that its value becomes observable only at time $\tau\in(0,T)$.
Therefore, the stochastic control problem can be divided into two sub-periods: $[0,\tau]$ and $(\tau,T]$. In analogy with the admissible set introduced in the main model, we define
\begin{equation*}
\mathcal A_1:=\left\{v_1:\Omega\times[0,\tau]\to\mathbb R \ \text{is } \mathcal F\text{-progressively measurable and } |v_1(t)|\le C \text{ on } [0,\tau]\times\Omega\right\},
\end{equation*}
and
\begin{equation*}
\mathcal A_2:=\left\{v_2:\Omega\times[\tau,T]\to\mathbb R \ \text{is } \mathcal F\text{-progressively measurable and } |v_2(t)|\le C \text{ on } [\tau,T]\times\Omega\right\}.
\end{equation*}
Here, $v_1$ represents the trading strategy adopted before the regulatory decision is revealed, whereas $v_2$ is the trading strategy adopted after time $\tau$, when the broker can condition on the observed value of $R$.

Prior to observing the regulatory decision, the broker assigns a subjective probability $p\in[0,1]$ to the event $R=1$ (regulatory approval). Since the outcome is revealed only at time $\tau\in(0,T)$, we decompose the problem into a post-decision and a pre-decision stage.

On the interval $[\tau,T]$, the continuation value depends on the realized outcome. We define
\begin{align*}
\mathcal J^{R=1}(t,x,q,S)
&=\sup_{v_2\in\mathcal A_2}
\mathbb E\!\left[-\exp\!\left(-\gamma Y^{\mathrm{Physical}}(T)\right)\right],\\
\mathcal J^{R=0}(t,x,q,S)
&=\sup_{v_2\in\mathcal A_2}
\mathbb E\!\left[-\exp\!\left(-\gamma Y^{\mathrm{Cash}}(T)\right)\right],
\end{align*}
for $t\in[\tau,T]$, the payoffs being defined in Equations (\ref{payoff_physical}) and (\ref{cash settlement}). 
From the perspective of any time $t<\tau$, the continuation value at time $\tau$ is evaluated in the expectation as
\begin{equation*}
\mathbb E\big[\mathcal J^\tau(\tau,x,q,S)\big]
=
p\,\mathcal J^{R=1}(\tau,x,q,S)
+
(1-p)\,\mathcal J^{R=0}(\tau,x,q,S).
\end{equation*}
Therefore, the pre-decision value function is given by
\begin{equation*}
\mathcal J^{\mathrm{pre}}(t,x,q,S)
=
\sup_{v_1\in\mathcal A_1}
\mathbb E\!\left[
p\,\mathcal J^{R=1}(\tau,X(\tau),Q(\tau),S(\tau))
+
(1-p)\,\mathcal J^{R=0}(\tau,X(\tau),Q(\tau),S(\tau))
\right],
\end{equation*}
for $t\in[0,\tau)$.

\begin{table}[t]
	\centering
    \caption{Fees of the contract for different values of $p$. }\label{Ext1}
	\begin{tabular}{|c|c|c|c|c|c|}
		\hline
$p$&	$p=0$(Cash)&	 $p=0.2$ & $p=0.5$ & $p=0.8$ &$p=1$(Physical)\\
		\hline
Price($\tau=0.5$, $\sigma=1$)&		   $45.0020$ & $45.0018$ & $45.0014$ &
$45.0010$ &$45.0007$ \\
\hline
Price($\tau=0.5$, $\sigma=5$)&		   $45.0130$ & $45.0112$ & $45.0081$ &
$45.0050$ &$45.0029$ \\
        \hline
	\end{tabular}
\end{table}
\begin{figure}[t]
    \centering
    \includegraphics[width=0.9\linewidth]{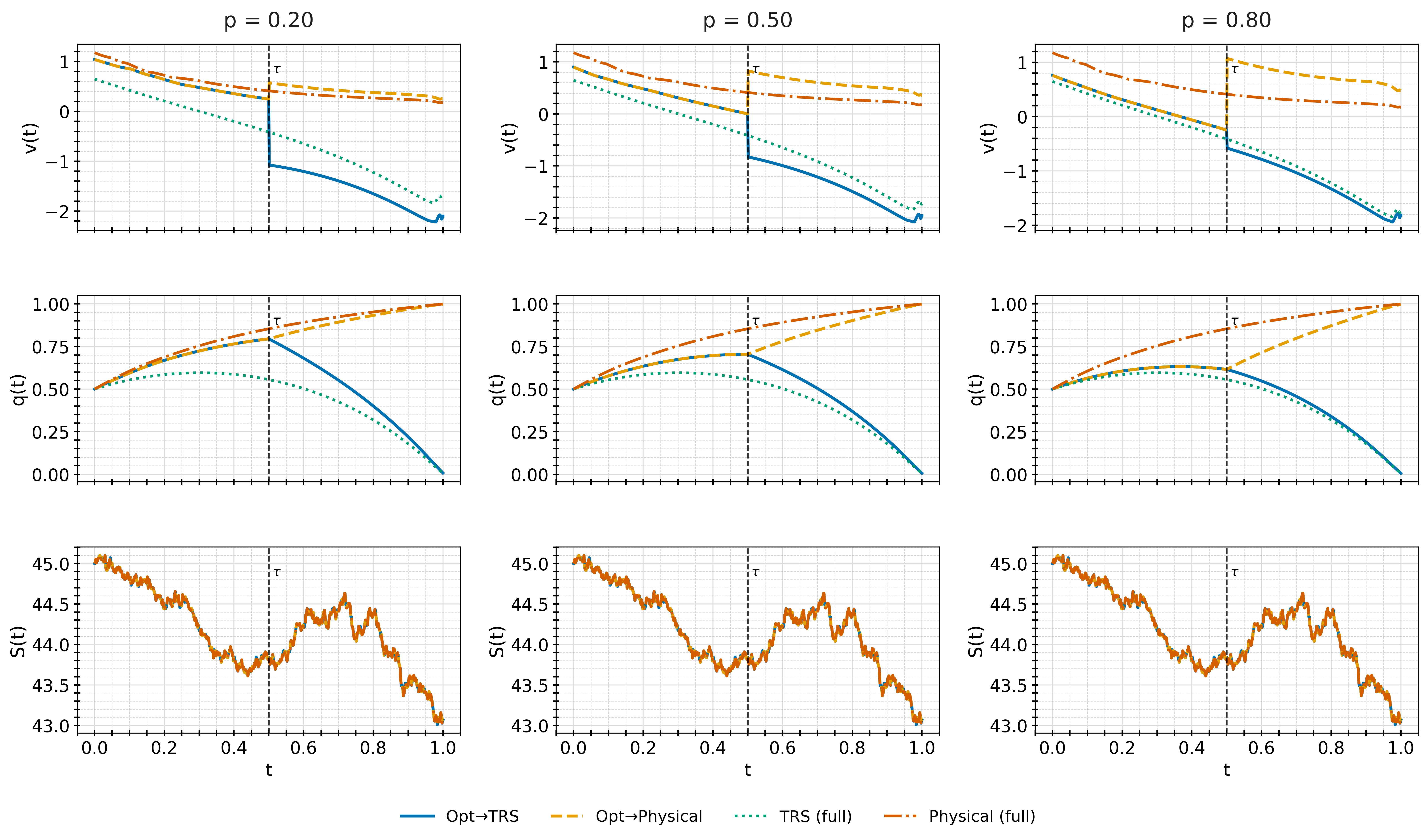}
    \caption{Simulation under different value of $p$ with $\tau=0.5$. $\sigma=1$.}
    \label{fig:placeholder4}
\end{figure}

Table \ref{Ext1} reports the fees as a function of $p$, setting $\tau=0.5=T/2$ and $\sigma=1$ and $5$, all the other parameters are as in Table \ref{parame0}. Notice that the fees correspond to the ones in Table \ref{tab:IP} for $p\in\{0,1\}$ and $\sigma=5$. As expected, increasing $p$, the fees decrease, reflecting the higher likelihood of ending up in the physical delivery regime, and therefore the less probable necessity to face high liquidation costs for the cash settlement cotnract.

In Figure \ref{fig:placeholder4} we report the trading strategies for different value of $p$ and $\sigma=1$ considering both the switch to  a cash-settled contract ($Opt \rightarrow TRS$) and to a physically delivered contract ($Opt \rightarrow Physical$). Notice that for $t \in [0, \tau]$ the strategy lies between the one obtained in case of physical delivery and cash-settled contract (TRS (full) and Physical (full), respectively). Results for $\sigma=5$ are not reported, since for large $\sigma$ the strategies for physical and TRS are very similar for $t<0.6$ (see Figure~\ref{fig:simPhys_TRS}), making  indistinguishable
the corresponding trajectories. 

\subsection{TWAP as a benchmark}
\label{TWAP}
Building on \cite{GUE15,JAI}, we also consider the case of a contract based on the TWAP, which is used in practice as a benchmark for delegated trading contracts, see \cite{BALD,BALD2,BAN,BIC,LARS}.

We introduce the TWAP process $A=(A_t)_{t\in[0,T]}$ defined as
\[
A_t=\frac{1}{t}\int_0^t S(u)\,du, \quad t>0,
\qquad A_0=S(0),
\]
and we modify the linear payoffs in Section \ref{MOD} setting the contract payoffs as follows:
\begin{equation*}
    Y^{TWAP}(T)=X(T)+Q(T)S(T)+N(A(T)-S(T))-L(Q(T)),
\end{equation*}
where $L(Q(T))=L(Q(T),0)$ in case of cash settlement and $L(Q(T))=L(Q(T),N)$ in case of physical delivery.

Therefore, we define the value function as follows
\begin{equation*}
    \mathcal{J}(t,x,q,S,A)=\sup_{v\in\mathcal{A}}\mathbb{E}\left[-\exp({-\gamma Y^{TWAP}(T)})|X(t)=x, Q(t)=q, S(t)=S, A(t)=A\right], t\leq T,
\end{equation*}
where we recall that $\gamma$ is the risk-aversion coefficient and $\mathcal{A}$ is the set of admissible strategies. Using the same arguments as in Section \ref{PROB}, the value function $\mathcal{J}(t,x,q,S,A)$ satisfies
\begin{equation*}
    \partial_t\mathcal{J}+\mu\partial_S\mathcal{J}+\frac{1}{2}\sigma^2\partial_{SS}\mathcal{J}+\frac{S-A}{t}\partial_A\mathcal{J}+\sup_{|v|\leq C}\{bv\partial_S\mathcal{J}+(rx-(S+lv)v)\partial_x\mathcal{J}+v\partial_q\mathcal{J}\}=0
\end{equation*}
with the terminal condition
\begin{equation*}
    \mathcal{J}(T,x,q,S,A)=-\exp({-\gamma}(x+qS+N(A-S)-L(q))).
\end{equation*}

This extension introduces an additional state variable through the TWAP process $A$. 
However, the dynamics remain of controlled diffusion type with continuous coefficients, and the payoff function is Lipschitz continuous. Therefore, the arguments developed in Section \ref{PROB} can be extended to this setting, yielding existence and uniqueness of the solution.

Equations (\ref{indifference_price}) and (\ref{transformation}) imply that
\begin{equation*}
    \mathcal{J}(t,x,q,S,A)=-\exp[-\gamma e^{r(T-t)}(x+qS-\mathcal{P}(t,q,S,A))].
\end{equation*}
By direct computation, the utility indifference fee $\mathcal{P}(t,q,S,A)$ satisfies
\begin{align*}
    -\partial_t\mathcal{P}-\frac{S-A}{t}\partial_A\mathcal{P}+r\mathcal{P}+(\mu-rS)q-\mu\partial_S\mathcal{P}-\frac{1}{2}\sigma^2\partial_{SS}\mathcal{P}-\frac{1}{2}\sigma^2\gamma e^{r(T-t)}(q-\partial_S\mathcal{P})^2\\+\sup_{|v|\leq C}\{-lv^2+(bq-b\partial_S\mathcal{P}-\partial_q\mathcal{P})v\}=0,
\end{align*}
with terminal condition
\begin{equation*}
    \mathcal{P}(T,q,S,A)=-N(A-S)+L(q).
\end{equation*}
To simplify the problem, we consider the case $r=0$ and apply the following transformation
\begin{equation*}
    \mathcal{P}(t,q,S,A)=N\frac{t}{T}\left(S-A\right)+U(t,q,S).
\end{equation*}
Then, we have
\begin{align*}
    &\partial_t\mathcal{P}=-\frac{N}{T}A+\frac{N}{T}S+\partial_tU,\quad \partial_S\mathcal{P}=N\frac{t}{T}+\partial_SU,\\
    &\partial_q\mathcal{P}=\partial_q U,\quad \partial_{SS}\mathcal{P}=\partial_{SS}U,\quad\partial_A\mathcal{P}=-N\frac{t}{T},
\end{align*}
and therefore
\begin{align*}
    -\partial_tU+\mu q-\mu\left(N\frac{t}{T}+\partial_SU\right)-\frac{1}{2}\sigma^2\partial_{SS}U-\frac{1}{2}\sigma^2\gamma \left(q-N\frac{t}{T}-\partial_SU\right)^2\\+\sup_{|v|\leq C}\left\{-lv^2+\left(bq-b\left(N\frac{t}{T}+\partial_SU\right)-\partial_q U\right)v\right\}=0,
\end{align*}
with terminal condition $U(T,q,S)=L(q)$. Notice that this transformation removes the dependence on the auxiliary state variable $A$, reducing the problem to a PDE in the variables $(t,q,S)$. 

\begin{table}[tp]
	\centering
    \caption{Fees of the two TWAP contracts assuming $q_0=0.5, S_0=45,\sigma=5$.}     \label{fees_TWAP}
	\begin{tabular}{|c|c|c|c|c|}
		\hline
		Model & TWAP(Physical) & TWAP(Cash)\\
		\hline
		Price & $0.4997$ & $0.4999$  \\
		\hline
	\end{tabular}
\end{table}

\begin{table}[tp]
	\centering
    \caption{Expected payoffs of the two TWAP cotnracts assuming $q_0=0.5, S_0=45,\sigma=5$.}\label{statarbTWAP}
	\begin{tabular}{|c|c|}
		\hline
		Model &  $\mathbb{E}[Y(T)|X(0)=X_0-q_0S_0+\mathcal{P}(0,q_0,S_0)]-X_0$ \\
		\hline
          TWAP(Physical) & 0.5117 \\
          TWAP(Cash) & 0.5631 \\
        \hline
	\end{tabular}
\end{table}
Fees are reported in Table \ref{fees_TWAP}. 
They are smaller than those of all previously considered contracts. 
This can be explained by comparing the payoffs. For instance, in the physical-delivery case, we have
\begin{equation*}
    Y^{{TWAP}}(T)=Y^{{Physical}}(T)+N A(T).
\end{equation*}
Since $A(T)$ is an average price, the additional term $N A(T)$ increases the broker's terminal payoff, thereby reducing the corresponding indifference fee. Table \ref{statarbTWAP} indicates a higher likelihood of statistical arbitrage opportunities compared to standard linear contracts, see Table \ref{STATARB}. Statistical arbitrage arises both in case of a cash-settled and physically-delivered contract.

In Figure \ref{fig: TWAP under sigma}, we report the optimal strategies and inventory over the interval $[0,T]$ for physical delivery, TRS, TWAP with physical settlement and cash settlement, considering different values of $\sigma$. All the other parameters are as in Table \ref{tab:params}, and $S_0=45.0$ and $q_0=0.5$.
We observe that, for low volatility, the optimal trading strategies under TWAP-based compensation are similar to those of the corresponding linear contracts. In contrast, when volatility is high, the optimal trading strategies under TWAP compensation differ significantly from the linear ones, they become more similar to each other at the beginning of the trading period (sale of the stock, small inventory) and diverge only when maturity approaches. 

When volatility is low, a TWAP compensation scheme renders a less (more) aggressive trading strategy at the beginning of the trading horizon in case of physical delivery (cash-settled contract), the reverse is observed when maturity approaches. The inventory for TWAP-contracts is always below the level observed for the linear contracts.

When volatility is high, the broker faces a high risk concerning the TWAP compensation part, this leads to a reduction of the inventory with a sale of the stock at the beginning of the tarding period. As maturity approaches, the average price becomes a sounder reference (less risk on the TWAP component) and the broker becomes more aggressive. Notice that a high volatility creates space for price manipulation both for cash-settled and physically delivered contracts.

\begin{figure}[t]
    \centering
    \includegraphics[width=0.9\linewidth]{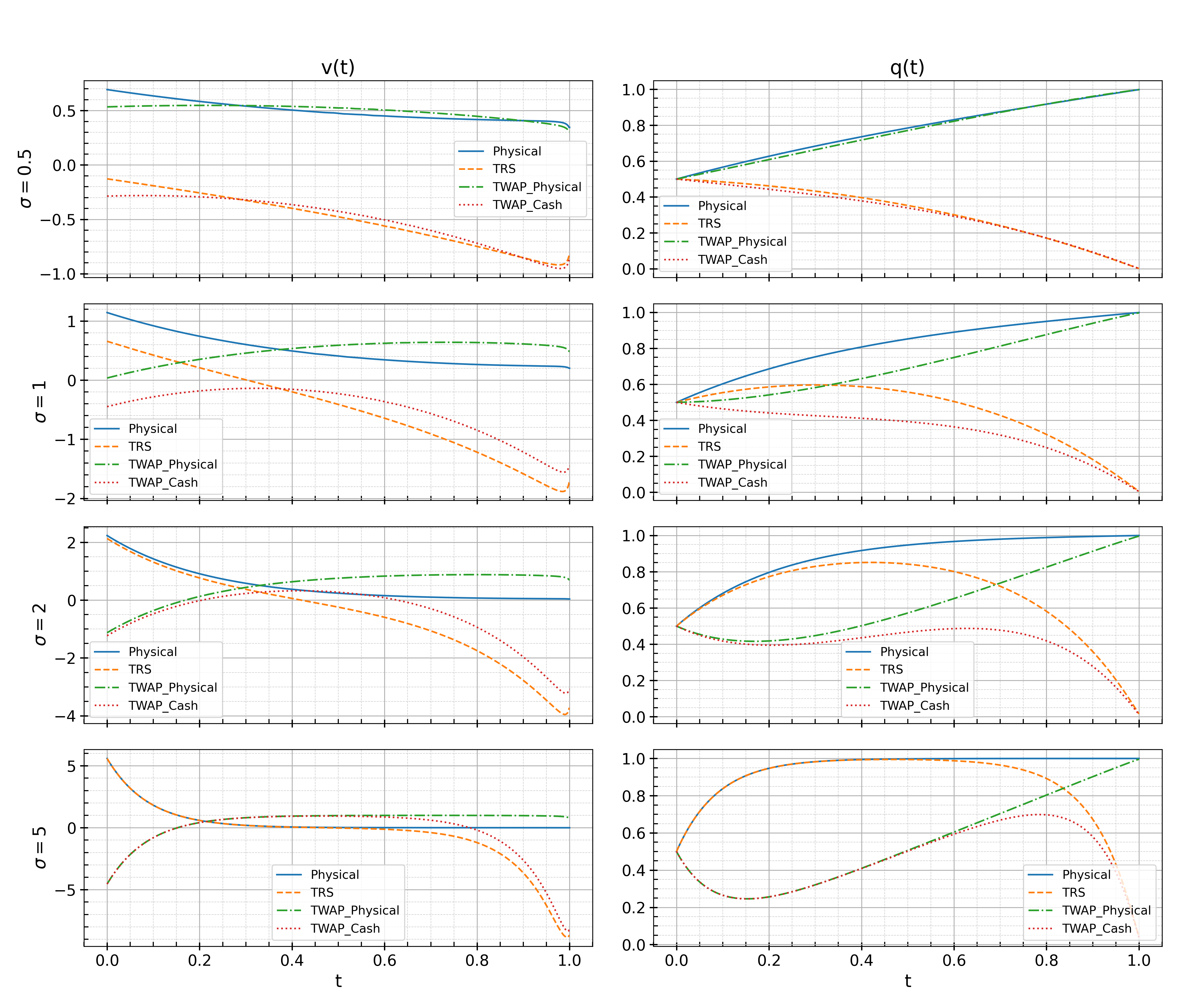}
    \caption{Comparison of the simulation of $v(t)$ and $Q(t)$ for the physical, TRS contracts, TWAP with physical and TWAP with cash contracts under different values of $\sigma$.}
    \label{fig: TWAP under sigma}
\end{figure}

\section{Conclusions}
In this paper we investigated the optimal execution of contracts that are used in merger\&acquisition deals. We considered cash-settled and physical delivered contracts between a broker and a counterpart. The contracts are linear (Total Returns Swaps), nonlinear (collar contracts) or Asian type (TWAP-based compensation). We derived the optimal execution strategy and the optimal fee through indifference utility arguments allowing for linear market effects. 

Two main results are provided. Cash-settled linear contracts, compared to physical delivery, are more expensive, are more exposed to manipulation (change in the sign of trades) and statistical arbitrages by the broker. This is due to the different terminal constraint on the stocks. 
Nonlinear contracts or TWAP-based contracts are structurally exposed to manipulation because hedging the stock-cash flow position strongly depends on price evolution. 

Further investigation calls for an analysis of these contracts in a game theoretic setting between the broker and the target company itself or another broker who can compete on the acquisition of the company.

\newpage

\appendix
\renewcommand{\thesection}{\Alph{section}}

\renewcommand{\thefigure}{\thesection\arabic{figure}}
\renewcommand{\thetable}{\thesection\arabic{table}}
\renewcommand{\thetheorem}{\thesection\arabic{theorem}}

\setcounter{figure}{0}
\setcounter{table}{0}
\setcounter{theorem}{0}

\section{Closed form solutions for $\mu=0$ and $r=0$}
\label{CLOSE}
Assuming $\mu=r=0$, we can obtain a closed form solution of the value function and of the indifference fee in case of a linear contract.

\begin{proposition}
\label{CFPhys}
    (physical delivery)
Assuming $\mu=r=0$, the value function for a physical delivery with payoff as in (\ref{payoff_physical}) is given by
\begin{equation*}
    \mathcal{J}_1^*(t,x,q,S)=-\exp(-\gamma (x+qS-NS-h(q,t))),
\end{equation*}
and the optimal strategy is
\begin{equation*}
    v^{*}_1(t,q)=\min\left\{C; \max\left\{-C;\frac{bq-bN-\partial_qh(t,q)}{2l}\right\}\right\},
\end{equation*}
where
\begin{equation*}
    h(t,q)=\left(a\frac{1+\xi e^{-\frac{2a}{l}(T-t)}}{1-\xi e^{-\frac{2a}{l}(T-t)}}+\frac{1}{2}b\right)(q-N)^2,
    \end{equation*} 
    with
    \begin{align*}
        a=\sqrt{\frac{l}{2}\sigma^2\gamma},\ \ \ 
        \xi=\frac{\alpha-\frac{b}{2}-a}{\alpha-\frac{b}{2}+a}.
    \end{align*}
\end{proposition}

\begin{proof}
    Assuming $\mu=r=0$, the HJB equation (\ref{HJB}) reduces to the following
\begin{equation}\label{HJB1_reduce}
    \partial_t\mathcal{J}_1+\frac{1}{2}\sigma^2\partial_{SS}\mathcal{J}_1+\sup_{|v|\leq C}\{bv\partial_S\mathcal{J}_1-(S+lv)v\partial_x\mathcal{J}_1+v\partial_q\mathcal{J}_1\}=0,
\end{equation}
with the terminal condition
\begin{equation*}
    \mathcal{J}_1(T,x,q,S)=-\exp(-\gamma(x+qS-NS-L(q))).
\end{equation*}
Assuming $\partial_x\mathcal{J}_1$ positive, the (unconstrained) supremum in Equation (\ref{HJB1_reduce}) is attained at
\begin{equation*}
    v^*=\frac{b\partial_S\mathcal{J}_1+\partial_q\mathcal{J}_1-S\partial_x\mathcal{J}_1}{2l\partial_x\mathcal{J}_1}.
\end{equation*}
Then the value function satisfies
\begin{equation}\label{HJB1_reduce2}
    \partial_t\mathcal{J}_1+\frac{1}{2}\sigma^2\partial_{SS}\mathcal{J}_1+\frac{(b\partial_S\mathcal{J}_1+\partial_q\mathcal{J}_1-S\partial_x\mathcal{J}_1)^2}{4l\partial_x\mathcal{J}_1}=0.
\end{equation}
Given the terminal time condition, we assume the following ansatz for the value function:
\begin{equation*}
    \mathcal{J}_1(t,x,q,S)=-\exp(-\gamma (x+qS-NS-h(q,t))).
\end{equation*}
By direct computation, we obtain
\begin{equation}\label{HJB1_reduce3}
    -\partial_th-\frac{1}{2}\sigma^2\gamma(q-N)^2+\frac{(bq-bN-\partial_qh)^2}{4l}=0
\end{equation}
and $h(T,q) = L(q, N)=\alpha(q-N)^2$.

This equation can be solved via separation of variables. 
Assuming $h(t,q)=\bar{h}(t)(q-N)^2$, $\bar{h}(t)$ satisfies
\begin{equation*}
    -\partial_t\bar{h}(t)-\frac{1}{2}\sigma^2\gamma+\frac{(b-2\bar{h}(t))^2}{4l}=0,
\end{equation*}
with terminal condition $\bar{h}(T)=\alpha$. The resulting ODE is a Riccati equation, which admits an explicit solution. To proceed, we first set $\bar{h}(t)=\theta(t)+\frac{1}{2}b$, and we obtain
\begin{equation*}
    \frac{\partial_t\theta}{\theta^2-\frac{l}{2}\sigma^2\gamma}=\frac{1}{l},
\end{equation*}
and terminal condition $\theta(T)=\alpha-\frac{b}{2}$. Integrating both sides over the interval $[t, T]$, we obtain
\begin{equation*}
    \theta(t)=a\frac{1+\xi e^{-\frac{2a}{l}(T-t)}}{1-\xi e^{-\frac{2a}{l}(T-t)}},
\end{equation*}
with
\begin{align*}
    a=\sqrt{\frac{l}{2}\sigma^2\gamma},\ \ \ 
    \xi=\frac{\alpha-\frac{b}{2}-a}{\alpha-\frac{b}{2}+a}.
\end{align*}
\end{proof}

\begin{proposition}\label{CFTRS}(TRS)
Assuming $\mu=r=0$, the value function for the TRF with payoff as in (\ref{cash settlement}) is given by
\begin{equation*}
    \mathcal{J}_2^*(t,x,q,S)=-\exp(-\gamma (x+qS-NS-\hat{h}(q,t))),
\end{equation*}
and the optimal strategy is
\begin{equation*}
    v^{*}_2(t,q)=\min\left\{C; \max\left\{-C;\frac{bq-bN-\partial_q\hat{h}(t,q)}{2l}\right\}\right\},
\end{equation*}
where
\begin{align*}
    \hat{h}(t,q)&=\hat{h}_0(t)+\hat{h}_1(t)q+\hat{h}_2(t)q^2.\\
    \hat{h}_2(t)&=a\frac{1+\xi e^{-\frac{2a}{l}(T-t)}}{1-\xi e^{-\frac{2a}{l}(T-t)}}+\frac{1}{2}b,\\
    \hat{h}_1(t)&=\left(\frac{e^{\frac{a}{l}(T-t)}-\xi e^{-\frac{a}{l}(T-t)}}{1-\xi}-1\right)bN\\
    &\quad\quad-\sigma^2\gamma N\left(\frac{e^{\frac{a}{l}(T-t)}-\xi e^{-\frac{a}{l}(T-t)}}{2\frac{a}{l}\sqrt{\xi}}\right)\ln\frac{(e^{\frac{a}{l}(T-t)}-\sqrt{\xi})(1+\sqrt{\xi})}{(e^{\frac{a}{l}(T-t)}+\sqrt{\xi})(1-\sqrt{\xi})},\\
    \hat{h}_0(t)&=\int_t^{T}\left[\frac{1}{4l}\left(\hat{h}_1(t)+bN\right)^2-\frac{1}{2}\sigma^2\gamma N^2\right]ds,
\end{align*}
where $a$ and $\xi$ are defined as in Proposition \ref{CFPhys}.
\end{proposition}
\begin{proof}
Following the same approach as for the proof of Proposition \ref{CFPhys}, we have $h(T,q)=L(q,0)=\alpha q^2$ and we assume that
\begin{equation*}
    \hat{h}(t,q)=\hat{h}_0(t)+\hat{h}_1(t)q+\hat{h}_2(t)q^2.
\end{equation*}
Direct computation shows that
\begin{align*}
    \partial_t\mathcal{J}_2&=-\gamma\partial_t\hat{h} e^{-\gamma (x+qS-NS-\hat{h}(q,t))} , \\
    \partial_S\mathcal{J}_2&=\gamma(q-N) e^{-\gamma (x+qS-NS-\hat{h}(q,t))}, \\
    \partial_{SS}\mathcal{J}_2&=-\gamma^2(q-N)^2 e^{-\gamma (x+qS-NS-\hat{h}(q,t))},\\
    \partial_x\mathcal{J}_2&=\gamma e^{-\gamma (x+qS-NS-\hat{h}(q,t))}, \\
    \partial_q\mathcal{J}_2&=\gamma(S-\partial_q \hat{h}) e^{-\gamma (x+qS-NS-\hat{h}(q,t))}.
\end{align*}
Substituting this expression into Equation (\ref{HJB1_reduce3}), we obtain that
\begin{equation*}
    -\partial_t\hat{h}_0-\partial_t\hat{h}_1q-\partial_t\hat{h}_2q^2-\frac{1}{2}\sigma^2\gamma(q-N)^2+\frac{(bq-bN-\hat{h}_1-2\hat{h}_2q)^2}{4l},
\end{equation*}
and terminal time condition $\hat{h}_0(T)=0$, $\hat{h}_1(T)=0$ and $\hat{h}_2(T)=\alpha$. By collecting the same terms in powers of $q$, we obtain the following system of ODEs:
\begin{align}
    &-\partial_t\hat{h}_2-\frac{1}{2}\sigma^2\gamma+\frac{(b-2\hat{h}_2)^2}{4l}=0, \label{ODEs_1}\\
    &-\partial_t\hat{h}_1+\sigma^2\gamma N-\frac{(b-2\hat{h}_2)(\hat{h}_1+bN)}{2l}=0, \label{ODEs_2}\\
    &-\partial_t\hat{h}_0-\frac{1}{2}\sigma^2\gamma N^2+\frac{(\hat{h}_1+bN)^2}{4l}=0. \label{ODEs_3}
\end{align}
We notice that Equation (\ref{ODEs_1}) is the same as the one in the proof of Proposition \ref{CFPhys}, and thus we obtain
\begin{equation*}
    \hat{h}_2(t)=a\frac{1+\xi e^{-\frac{2a}{l}(T-t)}}{1-\xi e^{-\frac{2a}{l}(T-t)}}+\frac{1}{2}b,
\end{equation*}
where 
\begin{align*}
    a=\sqrt{\frac{l}{2}\sigma^2\gamma},\ \ 
    \xi=\frac{\alpha-\frac{b}{2}-a}{\alpha-\frac{b}{2}+a}.
\end{align*}
To solve Equation (\ref{ODEs_2}), we apply the method of integrating factors. We first perform a change of variables
\begin{equation*}
    \bar{h}_1:=\hat{h}_1(t)+bN,
\end{equation*}
and the linear equation
\begin{equation*}
    -\partial_t\bar{h}_1+\sigma^2\gamma N+\frac{a}{l}\frac{1+\xi e^{-\frac{2a}{l}(T-t)}}{1-\xi e^{-\frac{2a}{l}(T-t)}}\bar{h}_1=0,
\end{equation*}
We set the integrating factors as
\begin{equation*}
    A(t)=\exp\left\{{-\frac{a}{l}\int_t^{T}\frac{1+\xi e^{-\frac{2a}{l}(T-s)}}{1-\xi e^{-\frac{2a}{l}(T-s)}}ds}\right\}=\frac{1-\xi}{e^{\frac{a}{l}(T-t)}-\xi e^{-\frac{a}{l}(T-t)}}.
\end{equation*}
We use the identity $\partial_t(A(t)\bar{h}_1(t))=A(t)\sigma^2\gamma N$, and integrate over $[t,T]$, obtaining
\begin{equation*}
    \bar{h}_1(t)=A(t)^{-1}\left(A(T)bN-\sigma^2\gamma N\int_t^{T}A(s)ds\right).
\end{equation*}
Combine with the terminal time condition $\hat{h}_1(T)=0$, we obtain
\begin{align*}
    \hat{h}_1(t)&=\left(\frac{e^{\frac{a}{l}(T-t)}-\xi e^{-\frac{a}{l}(T-t)}}{1-\xi}-1\right)bN\\
    &-\sigma^2\gamma N\left(\frac{e^{\frac{a}{l}(T-t)}-\xi e^{-\frac{a}{l}(T-t)}}{2\frac{a}{l}\sqrt{\xi}}\right)\ln\frac{(e^{\frac{a}{l}(T-t)}-\sqrt{\xi})(1+\sqrt{\xi})}{(e^{\frac{a}{l}(T-t)}+\sqrt{\xi})(1-\sqrt{\xi})}
\end{align*}
where $a$ and $\xi$ are defined as in the proof of Proposition \ref{CFPhys}.
Substituting $\hat{h}_1(t)$ into Equation (\ref{ODEs_3}), and considering the terminal time condition $\hat{h}_0(T)=0$, we obtain
\begin{equation*}
    \hat{h}_0(t)=\int_t^{T}\left[\frac{1}{4l}\left(\hat{h}_1(t)+bN\right)^2-\frac{1}{2}\sigma^2\gamma N^2\right]ds,
\end{equation*}
which can be evaluated numerically.
\end{proof}

Thanks to the relationship between $\mathcal{J}(t,x,q,S)$ and $\mathcal{P}(t,q,S)$ in Equation (\ref{fromJto_P}):
\begin{equation*}
    \mathcal{J}(t,x,q,S)=-\exp[-\gamma e^{r(T-t)}(x+qS-\mathcal{P}(t,q,S))],
\end{equation*}
we can obtain the closed-form expression of $\mathcal{P}(t,q,S)$ for the two contracts. By direct computation, we obtain the following corollaries.

\begin{corollary}
\label{PHYSICALprice}
    (physical delivery)
Assuming $\mu=r=0$, the utility indifference fee for the payoff in (\ref{payoff_physical}) is given by
\begin{equation}
    \mathcal{P}_1^*(t,q,S)=NS+h(q,t),
\end{equation}
where
\begin{equation}
    h(t,q)=\left(a\frac{1+\xi e^{-\frac{2a}{l}(T-t)}}{1-\xi e^{-\frac{2a}{l}(T-t)}}+\frac{1}{2}b\right)(q-N)^2.
    \end{equation} 
\end{corollary}

\begin{corollary}(TRS)
\label{TRSprice}
Assuming $\mu=r=0$, the utility indifference fee for the payoff in (\ref{cash settlement}) is given by
\begin{align*}
    &\mathcal{P}_2^*(t,q,S)=NS+\hat{h}(q,t),\\
    &\hat{h}(t,q)=\hat{h}_0(t)+\hat{h}_1(t)q+\hat{h}_2(t)q^2.
\end{align*}
where
\begin{align*}
    \hat{h}_2(t)&=a\frac{1+\xi e^{-\frac{2a}{l}(T-t)}}{1-\xi e^{-\frac{2a}{l}(T-t)}}+\frac{1}{2}b,\\
    \hat{h}_1(t)&=\left(\frac{e^{\frac{a}{l}(T-t)}-\xi e^{-\frac{a}{l}(T-t)}}{1-\xi}-1\right)bN\\
    &\quad \quad-\sigma^2\gamma N\left(\frac{e^{\frac{a}{l}(T-t)}-\xi e^{-\frac{a}{l}(T-t)}}{2\frac{a}{l}\sqrt{\xi}}\right)\ln\frac{(e^{\frac{a}{l}(T-t)}-\sqrt{\xi})(1+\sqrt{\xi})}{(e^{\frac{a}{l}(T-t)}+\sqrt{\xi})(1-\sqrt{\xi})},\\
    \hat{h}_0(t)&=\int_t^{T}\left[\frac{1}{4l}\left(\hat{h}_1(t)+bN\right)^2-\frac{1}{2}\sigma^2\gamma N^{2}\right]ds.
\end{align*}

\end{corollary}

\section{Numerical method}\label{Numerical method}
As defined in Section~\ref{section-utility indifference fee}, we apply numerical methods to the utility indifference fee $\mathcal{P}(t, q, S)$ rather than to the value function $\mathcal{J}(t, x, q, S)$, since the former reduces the dimensionality and thus significantly simplifies the computational complexity. 

The HJB for the utility indifference fee is:
\begin{equation*}
    -\partial_t\mathcal{P}+r\mathcal{P}+(\mu-rS)q-\mu\partial_S\mathcal{P}-\frac{1}{2}\sigma^2\partial_{SS}\mathcal{P}-\frac{1}{2}\sigma^2\gamma e^{r(T-t)}(q-\partial_S\mathcal{P})^2+\sup_{|v|\leq C}\{-lv^2+(bq-b\partial_S\mathcal{P}-\partial_q\mathcal{P})v\}=0
\end{equation*}
satisfying the terminal condition
\begin{equation*}
    \mathcal{P}(T,q,S)=\Pi(S)+L(q).
\end{equation*}
At this point, infinitesimal generator attains its supremum at 
$$v^*=\min\left\{C; \max\left\{-C;\frac{bq-b\partial_S\mathcal{P}-\partial_q\mathcal{P}}{2l}\right\}\right\}.$$
Assuming that 
$$\frac{bq-b\partial_S\mathcal{P}-\partial_q\mathcal{P}}{2l}\in [-C,C],$$ then the HJB satisfied
\begin{equation*}
    -\partial_t\mathcal{P}+r\mathcal{P}+(\mu-rS)q-\mu\partial_S\mathcal{P}-\frac{1}{2}\sigma^2\partial_{SS}\mathcal{P}-\frac{1}{2}\sigma^2\gamma e^{r(T-t)}(q-\partial_S\mathcal{P})^2+\frac{(bq-b\partial_S\mathcal{P}-\partial_q\mathcal{P})^2}{4l}=0.
\end{equation*}
To solve this stochastic control problem numerically, we employ a semi-implicit operator splitting scheme:
\begin{itemize}
    \item For the linear part, $r\mathcal{P}+(\mu-rS)q-\mu\partial_S\mathcal{P}-\frac{1}{2}\sigma^2\partial_{SS}\mathcal{P}$, we employ an implicit FDM(Finite difference scheme).
    \item For the nonlinear part, $-\frac{1}{2}\sigma^2\gamma e^{r(T-t)}(q-\partial_S\mathcal{P})^2+\frac{(bq-b\partial_S\mathcal{P}-\partial_q\mathcal{P})^2}{4l}$, we adopt a monotonic explicit FDM.
\end{itemize}

The problem is on the domain
\begin{equation*}
    (t,S,q)\in[0,T]\times[-\infty,+\infty]\times[q_{min},q_{max}]. 
\end{equation*}
To implement the numerical method, we localize the domain to
\begin{equation*}
    (t,S,q)\in[0,T]\times[S_{min},S_{max}]\times[q_{min},q_{max}].
\end{equation*}

We will set $q_{max}$ sufficiently large to ensure that the search for the optimal control can be within the domain. No additional economic boundary condition is imposed in the inventory variable. At 
$q_{\min}$ and $q_{\max}$ one-sided finite differences are used. We also set $\partial_{SS}\mathcal{P}=0$ at the boundaries $S_{\min},S_{\max}$.

Now we begin the discretization of the equation (\ref{indifferencepriceHJB}). First, we define an equally spaced grid $S_i=S_{\min}+i\Delta S, i=0,...,I, \Delta S=(S_{\max}-S_{\min})/I$ and $q_j=q_{\min}+j\Delta q, j=0,...,J, \Delta q=(q_{\max}-q_{\min})/J$. The discrete time grid is defined by $t_n=n\Delta t$, where $\Delta t=T/N$ denotes the uniform time step size.

We define the two components mentioned previously as two operators. More specifically:
\begin{align*}
    \text{Linear operator: }L^1\mathcal{P}&:=r\mathcal{P}+(\mu-rS)q-\mu\partial_S\mathcal{P}-\frac{1}{2}\sigma^2\partial_{SS}\mathcal{P}\\
    \text{Nonlinear operator: }L^2\mathcal{P}&:=-\frac{1}{2}\sigma^2\gamma e^{r(T-t)}(q-\partial_S\mathcal{P})^2+\frac{(bq-b\partial_S\mathcal{P}-\partial_q\mathcal{P})^2}{4l}
\end{align*}
Let $\mathcal{P}(t_n,S_i,q_j)$ denote the exact solution of equation evaluated at the grid points $(t_n, S_i, q_j)$, and $\mathcal{P}^n_{i,j}$ the corresponding discrete approximation at point $(t_n, S_i, q_j)$. Moreover $(L^1\mathcal{P})^n_{i,j}$ and $(L^2\mathcal{P})^n_{i,j}$ denote the discrete value of the differential operator $L^1\mathcal{P}$ and $L^2\mathcal{P}$ at point $(t_n, S_i, q_j)$.

The non-linear part of $L^2\mathcal{P}$ is discretized explicitly as
\begin{align*}
	(L^2\mathcal{P})^{n+1}_{i,j}&=-\frac{1}{2}\sigma^2\gamma e^{r(T-(n+1)\Delta t)}\left(q_j-\frac{\mathcal{P}^{n+1}_{i+1,j}-\mathcal{P}^{n+1}_{i-1,j}}{2\Delta S}\right)^2\\&+\frac{1}{4l}\left(bq_j-b\frac{\mathcal{P}^{n+1}_{i+1,j}-\mathcal{P}^{n+1}_{i-1,j}}{2\Delta S}-\frac{\mathcal{P}^{n+1}_{i,j+1}-\mathcal{P}^{n+1}_{i,j-1}}{2\Delta q}\right)^2.
\end{align*}
The linear part of $L^1\mathcal{P}$ is discretized implicitly as
\begin{align*}
	(L^1\mathcal{P})^n_{i,j}&=r\mathcal{P}^{n}_{i,j}+(\mu-rS_i)q_j-\mu\frac{\mathcal{P}^{n}_{i+1,j}-\mathcal{P}^{n}_{i-1,j}}{2\Delta S}-\frac{1}{2}\sigma^2\left(\frac{\mathcal{P}^{n}_{i+1,j}-2\mathcal{P}^{n}_{i,j}+\mathcal{P}^{n}_{i-1,j}}{(\Delta S)^2}\right).
\end{align*}
Therefore, the final approximation scheme reads as follows
\begin{align*}
\frac{\mathcal{P}^{n+1}_{i,j}-\mathcal{P}^{n}_{i,j}}{\Delta t}&=(L^1\mathcal{P})^n_{i,j}+(L^2\mathcal{P})^{n+1}_{i,j}.
\end{align*}
Moreover, we can write it in matrix form as
\begin{align*}
\left(\frac{\sigma^2\Delta t}{2(\Delta S)^2}+\frac{\mu\Delta t}{2\Delta S}\right)\mathcal{P}^{n}_{i+1,j}&-\left(\frac{\sigma^2\Delta t}{(\Delta S)^2}+1+r\Delta t\right)\mathcal{P}^{n}_{i,j}+\left(\frac{\sigma^2\Delta t}{2(\Delta S)^2}-\frac{\mu\Delta t}{2\Delta S}\right)\mathcal{P}^{n}_{i-1,j}\\&=-\mathcal{P}^{n+1}_{i,j}-\frac{\Delta t}{2}\sigma^2\gamma e^{r(T-({n+1})\Delta t)}\left(q_j-\frac{\mathcal{P}^{n+1}_{i+1,j}-\mathcal{P}^{n+1}_{i-1,j}}{2\Delta S}\right)^2\\&+\frac{\Delta t}{4l}\left(bq_j-b\frac{\mathcal{P}^{n+1}_{i+1,j}-\mathcal{P}^{n+1}_{i-1,j}}{2\Delta S}-\frac{\mathcal{P}^{n+1}_{i,j+1}-\mathcal{P}^{n+1}_{i,j-1}}{2\Delta q}\right)^2+(\mu-rS_i)q_j\Delta t,
\end{align*}
with the terminal condition
\begin{equation*}
    \mathcal{P}^{N}_{i,j}=\Pi(S_i)+L(q_j),
\end{equation*}
and boundary setting
\begin{align*}
    &\partial_S\mathcal{P}^n_{0,j}=\partial_S\mathcal{P}^n_{1,j},\\
    &\partial_S\mathcal{P}^n_{I,j}=\partial_S\mathcal{P}^n_{I-1,j},\\
    &\partial_q\mathcal{P}^n_{i,0}=\frac{\mathcal{P}^{n}_{i,1}-\mathcal{P}^{n}_{i,0}}{\Delta q},\\
    &\partial_q\mathcal{P}^n_{i,J}=\frac{\mathcal{P}^{n}_{i,J}-\mathcal{P}^{n}_{i,J-1}}{\Delta q}.\\
\end{align*}
\begin{algorithm}[tp]
	\caption{Numerical method for the PDE}
	\begin{algorithmic}[1]
		\Require Terminal condition $\mathcal{P}^{N}_{i,j}$
		\Function{solve\_pde}{$T, Domain, Grids, Parameters$}
		\For{$n=N-1,\cdots,0$}
		\State $A \gets \text{Build\_Matrix}$ \Comment{Calculate matrix using implicit scheme}
		\State $O\_nlin \gets \text{Explicit\_Step}(\mathcal{P}[n+1])$ \Comment{Compute explicit step}
		\State $b \gets O\_nlin \times \Delta t + \mathcal{P}[n+1]$ \Comment{Compute right-hand side of equation}
		\State $\mathcal{P}[n] \gets \text{sparse\_solver}(A, b)$ \Comment{Update value function}
		\State $v[n] \gets \text{Optimal\_trading}(\mathcal{P}[n])$ \Comment{Calculate optimal strategy}
		\EndFor
		\State \Return $\mathcal{P}, v$
		\EndFunction
	\end{algorithmic}
\end{algorithm}
In our numerical experiments, we set the discretization parameters as in Table \ref{tab:paramsN}.
\begin{table}[H]
	\centering
    \caption{Discretization parameters}
    \label{parame1}
	\begin{tabular}{|c|c|}
		\hline
		Parameter & Value  \\
		\hline
        $N$ & 1001\\
        $J$ & 101 \\
		$I$ & 101 \\
		$S_{min}$ & 15  \\
		$S_{max}$ & 75 \\
        $q_{min}$ & -1 \\
		$q_{max}$ & 1 \\
		\hline
	\end{tabular}\label{tab:paramsN}
\end{table}.
\clearpage
\section{Additional Figures and Tables}
\label{additional}

In Figures \ref{fig:phys}-\ref{fig:col2} we report the indifference fee and the optimal trading strategy for the four contracts at time $t=0, \ 0.5$ and $0.95$. 

\begin{figure}[h]
    \centering
    \includegraphics[width=0.8\linewidth]{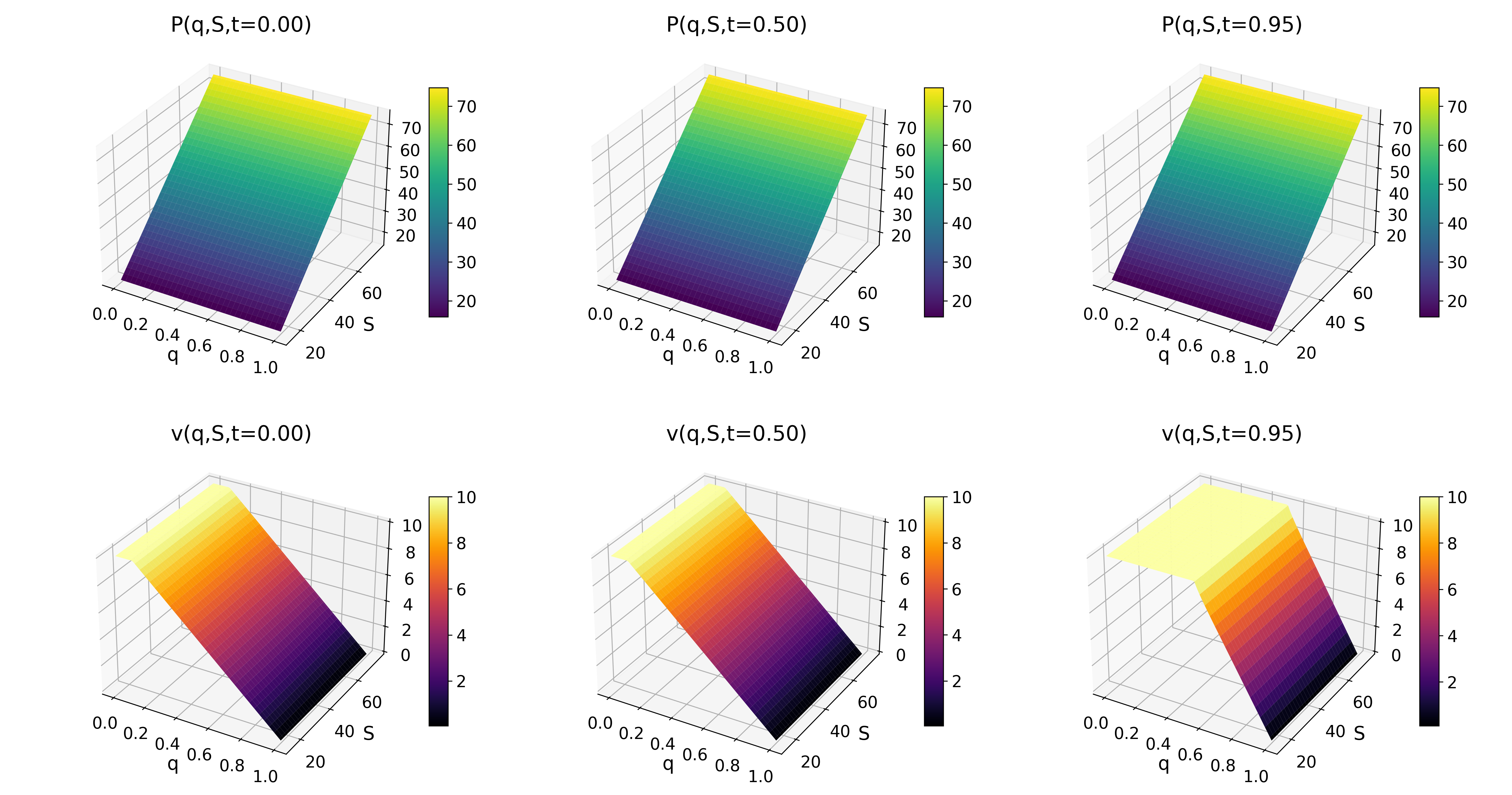}
    \caption{Physical delivery contract: indifference fee $\mathcal{P}$ (above) and optimal control (below) at $t=0,\ 0.5,\ 0.95$}
    \label{fig:phys}
\end{figure}
\begin{figure}[h]
    \centering
    \includegraphics[width=0.8\linewidth]{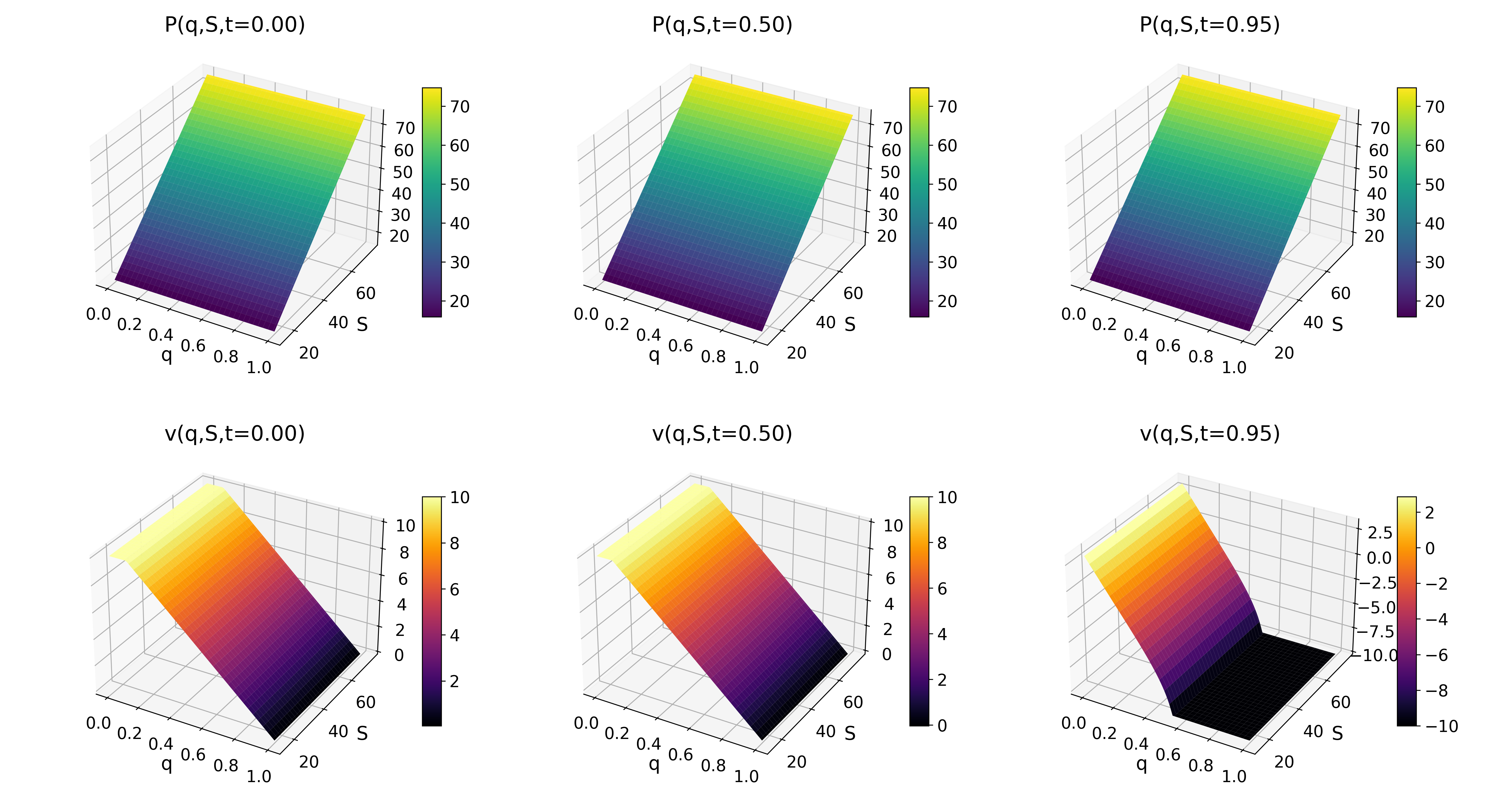}
    \caption{Cash settlement contract: indifference fee $\mathcal{P}$ (above) and the optimal control (below) at $t=0.0,\ 0.5,\ 0.95$}
    \label{fig:TRS}
\end{figure}
\begin{figure}[h]
    \centering
    \includegraphics[width=0.9\linewidth]{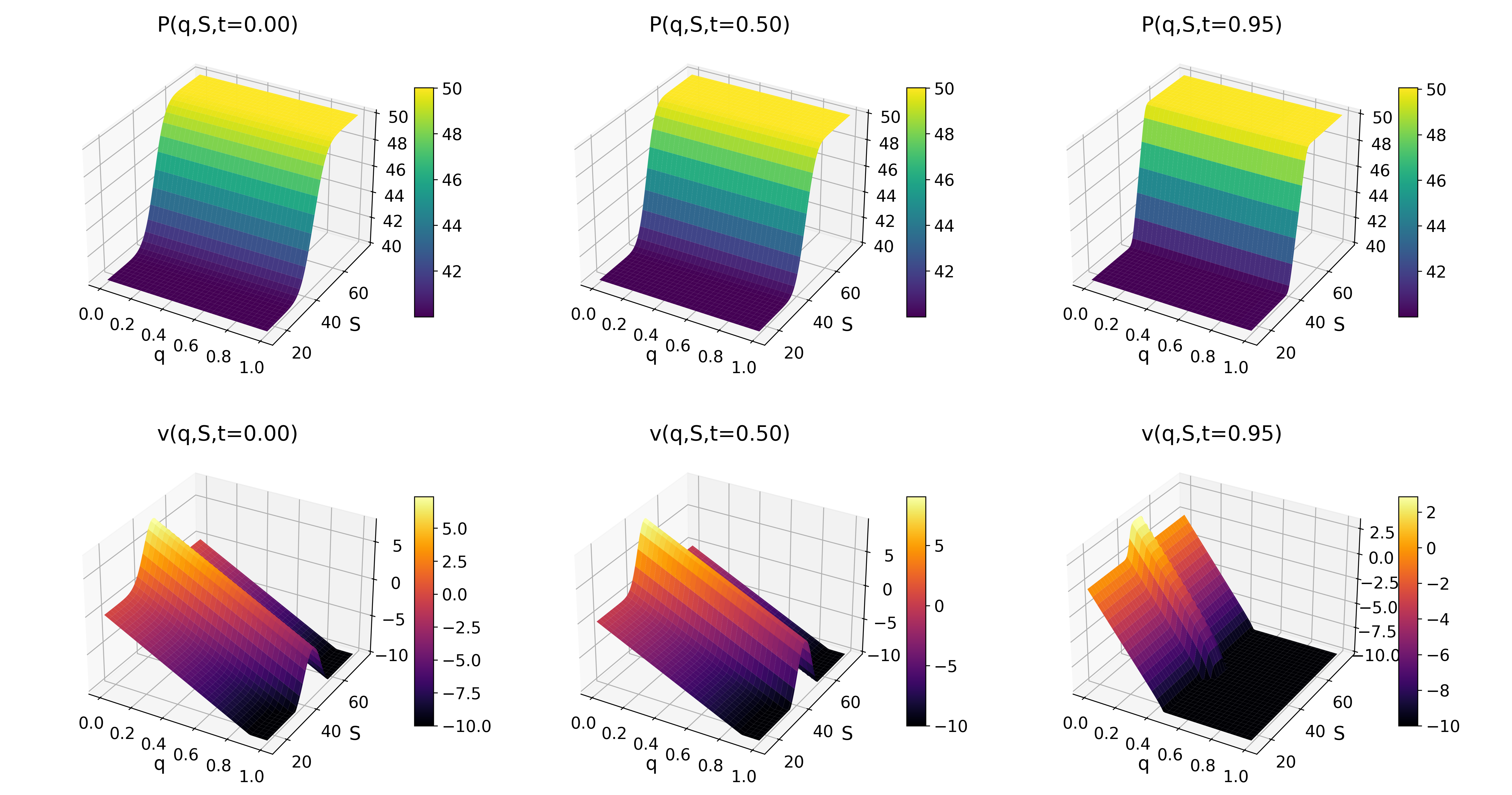}
    \caption{Collar cash contract: indifference fee $\mathcal{P}$ (above) and the optimal control (below) at $t=0,\ 0.5,\ 0.95$}
    \label{fig:col}
\end{figure}

\begin{figure}[h]
    \centering
    \includegraphics[width=0.9\linewidth]{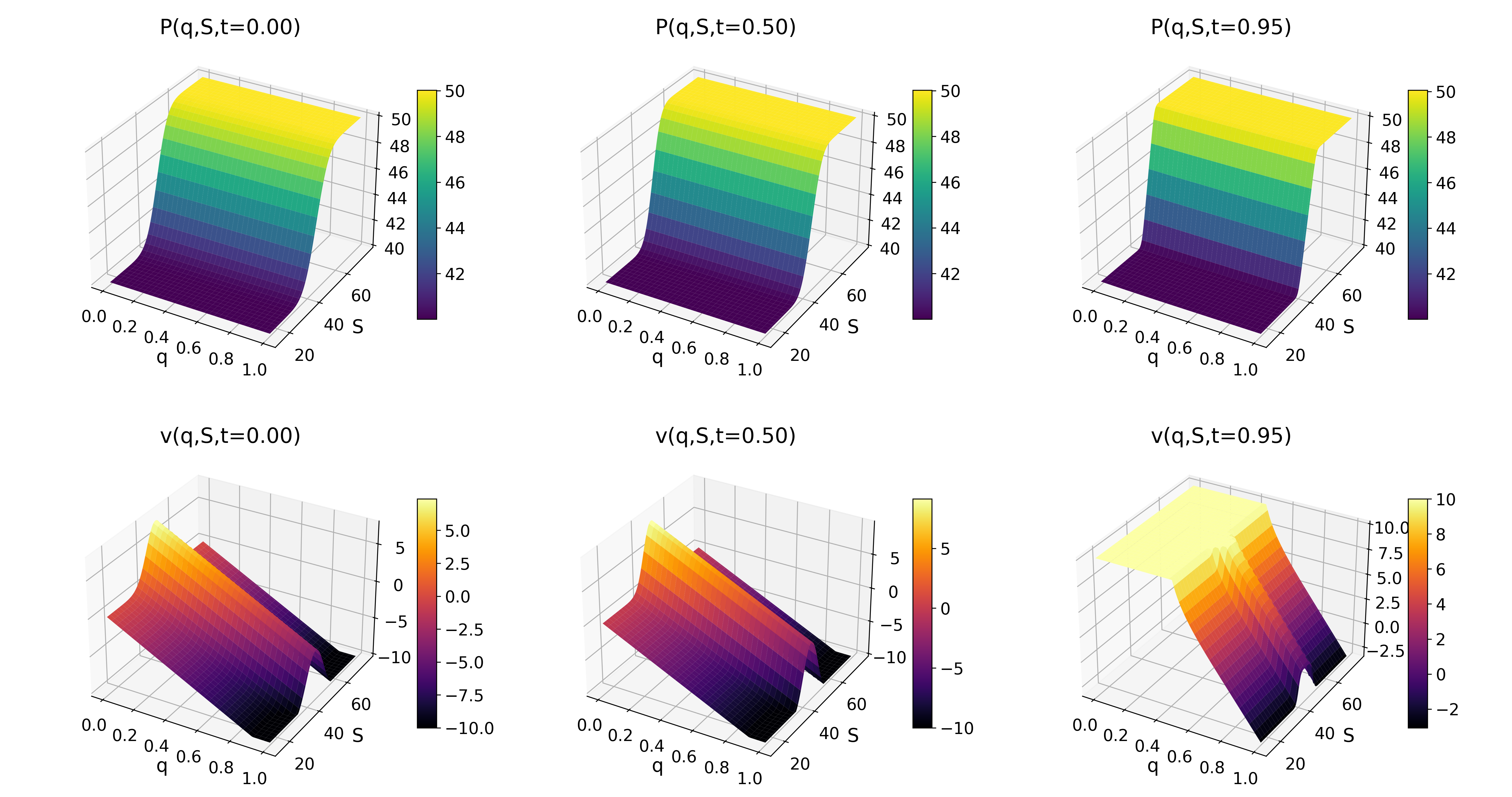}
    \caption{Collar physical contract: indifference fee $\mathcal{P}$ (above) and the optimal control (below) at $t=0,\ 0.5,\ 0.95$}
    \label{fig:col2}
\end{figure}
\clearpage
\subsection{Sensitivity analysis}\label{additionalSA}
In what follows we provide a sensitivity analysis considering the simulated scenario of Figure \ref{fig:simPhys_Collar}(left). When a parameter is swept, all the other parameters are as in Table \ref{tab:params}. All $4×2$ figure are divided into three parts: the first row corresponds to physical delivery contract, the second to the TRS contract, the third to collar with cash settlement, and the fourth to collar with physical delivery. 

\begin{figure}[h]
    \centering
    \includegraphics[width=0.9\linewidth]{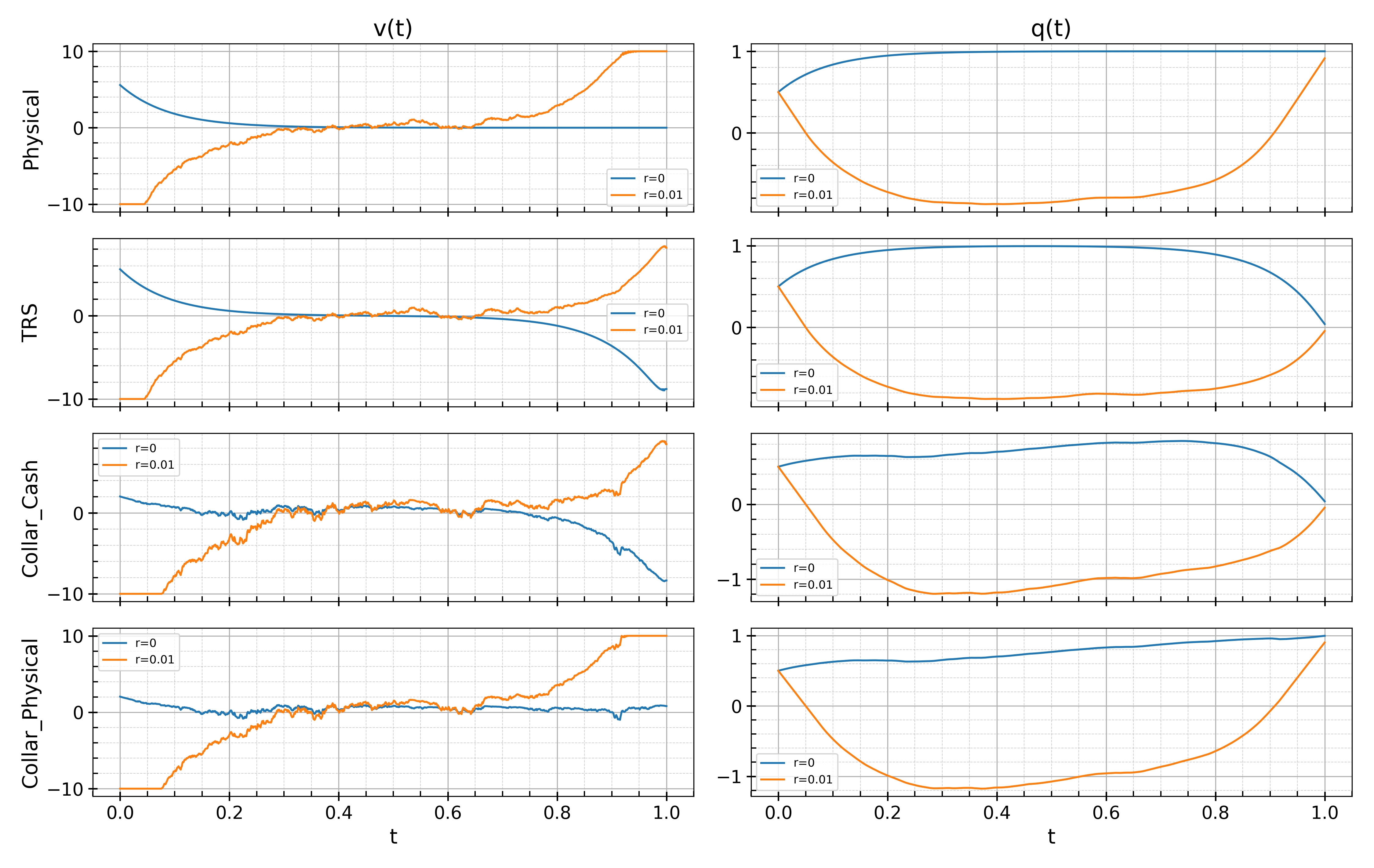}
    \caption{Simulations for different values of $r$.}
    \label{fig:sweep_r}
\end{figure}

\begin{table}[hp]
	\centering
    \caption{Fees for the four contracts}
	\begin{tabular}{|c|c|c|c|c|}
		\hline
		Parameter/Model & Physical & TRS  & Collar(Physical) & Collar(Cash)\\
		\hline
		  $r$=0.0 & $45.0029$ & $45.0130$  & $45.0042$ & $45.0078$\\
        $r$=0.01 & $44.6504$ & $44.6191$  & $44.5408$ &$44.4986$\\
		\hline
        $\mu$=-0.5 & $44.5735$ & $44.5362$  & $44.6128$ & $44.5632$\\
        $\mu$=0.0 & $45.0029$ & $45.0130$  & $45.0042$ &$45.0078$\\
        $\mu$=0.5 & $44.6346$ & $44.7134$  & $44.5973$ &$44.6617$\\
        \hline
		  $\sigma$=5 & $45.0029$ & $45.0130$ & $45.0042$ & $45.0078$\\
        $\sigma$=6 & $45.0034$ & $45.0157$ & $45.0059$ &$45.0082$ \\
        $\sigma$=7 & $45.0040$ & $45.0185$ & $45.0079$ &$45.0086$ \\
		\hline
        $\gamma$=0.001 & $45.0010$ & $45.0038$ & $45.0009$ & $45.0020$ \\
        $\gamma$=0.005 & $45.0021$ & $45.0092$ & $45.0027$ &$45.0053$ \\
        $\gamma$=0.01 & $45.0029$ & $45.0130$ & $45.0042$ &$45.0078$ \\
		\hline
        $l$=0.001 & $45.0029$ & $45.0130$ & $45.0042$ & $45.0078$ \\
        $l$=0.002 & $45.0041$ & $45.0184$ & $45.0054$ &$45.0107$ \\
        $l$=0.003 & $45.0049$ & $45.0223$ & $45.0062$ &$45.0126$ \\
        \hline
        $\alpha$=0.002 & $45.0029$ & $45.0046$ & $45.0024$ &$45.0030$ \\
        $\alpha$=0.02 & $45.0029$ & $45.0099$ & $45.0035$ &$45.0061$ \\
        $\alpha$=0.2 & $45.0029$ & $45.0130$ & $45.0042$ & $45.0078$ \\
        \hline
	\end{tabular}
\label{fees}
\end{table}

\begin{figure}[h]
    \centering
    \includegraphics[width=0.9\linewidth]{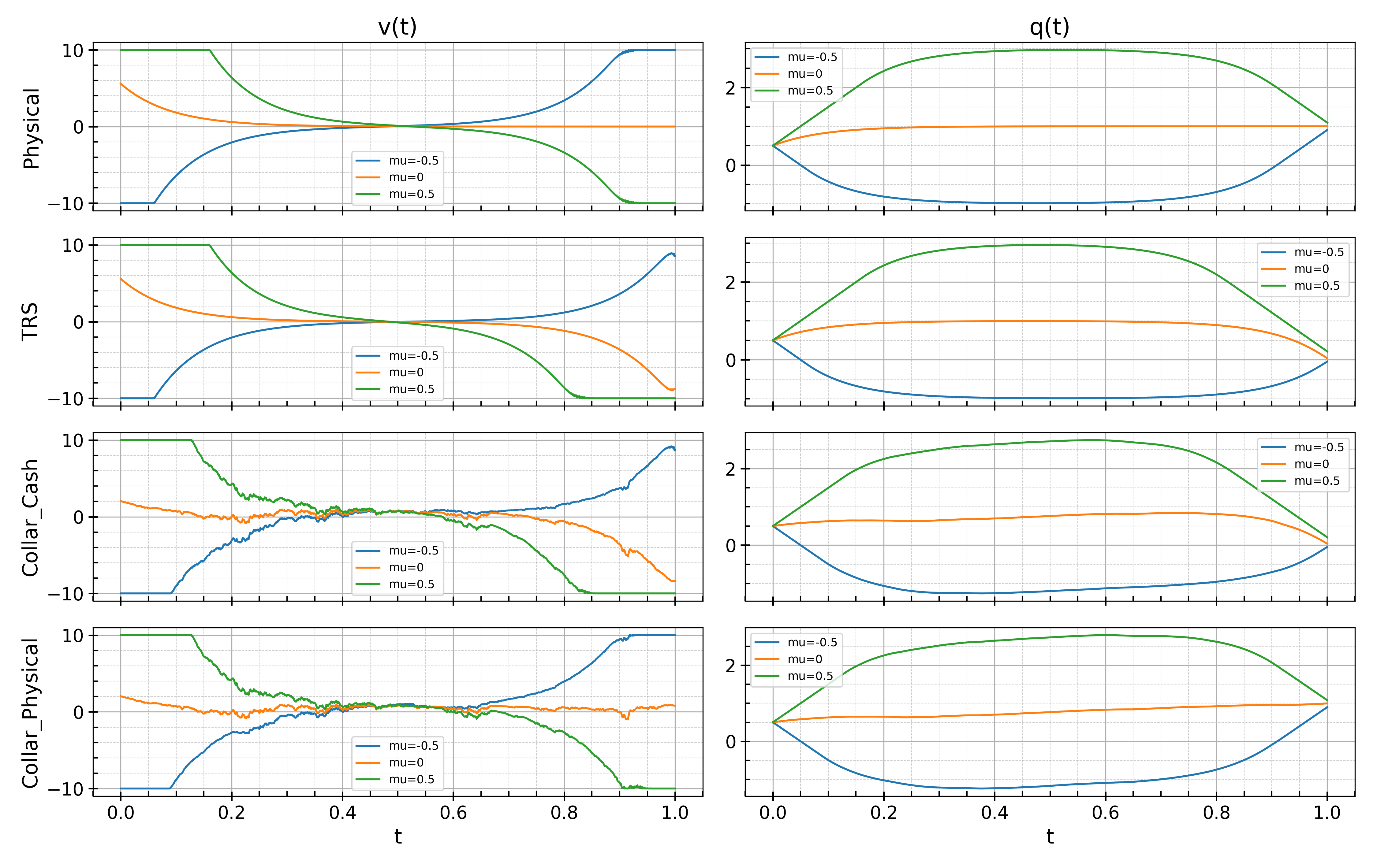}
    \caption{Simulations for different values of $\mu$.}
    \label{fig:sweep_mu}
\end{figure}

\begin{figure}[h]
    \centering
    \includegraphics[width=0.9\linewidth]{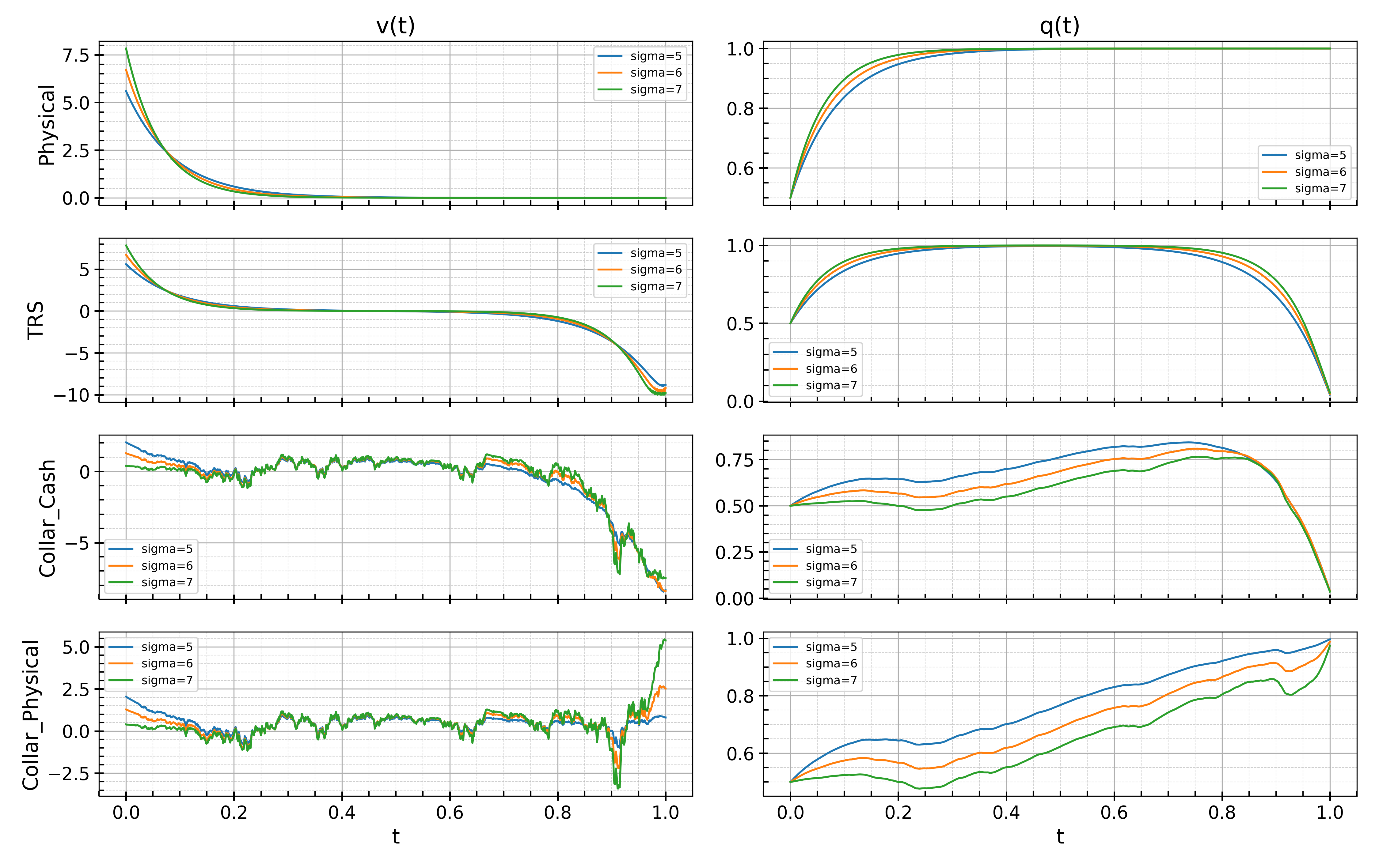}
    \caption{Simulations for different values of $\sigma$}
    \label{fig:sweep_sigma}
\end{figure}

\begin{figure}[h]
    \centering
    \includegraphics[width=0.9\linewidth]{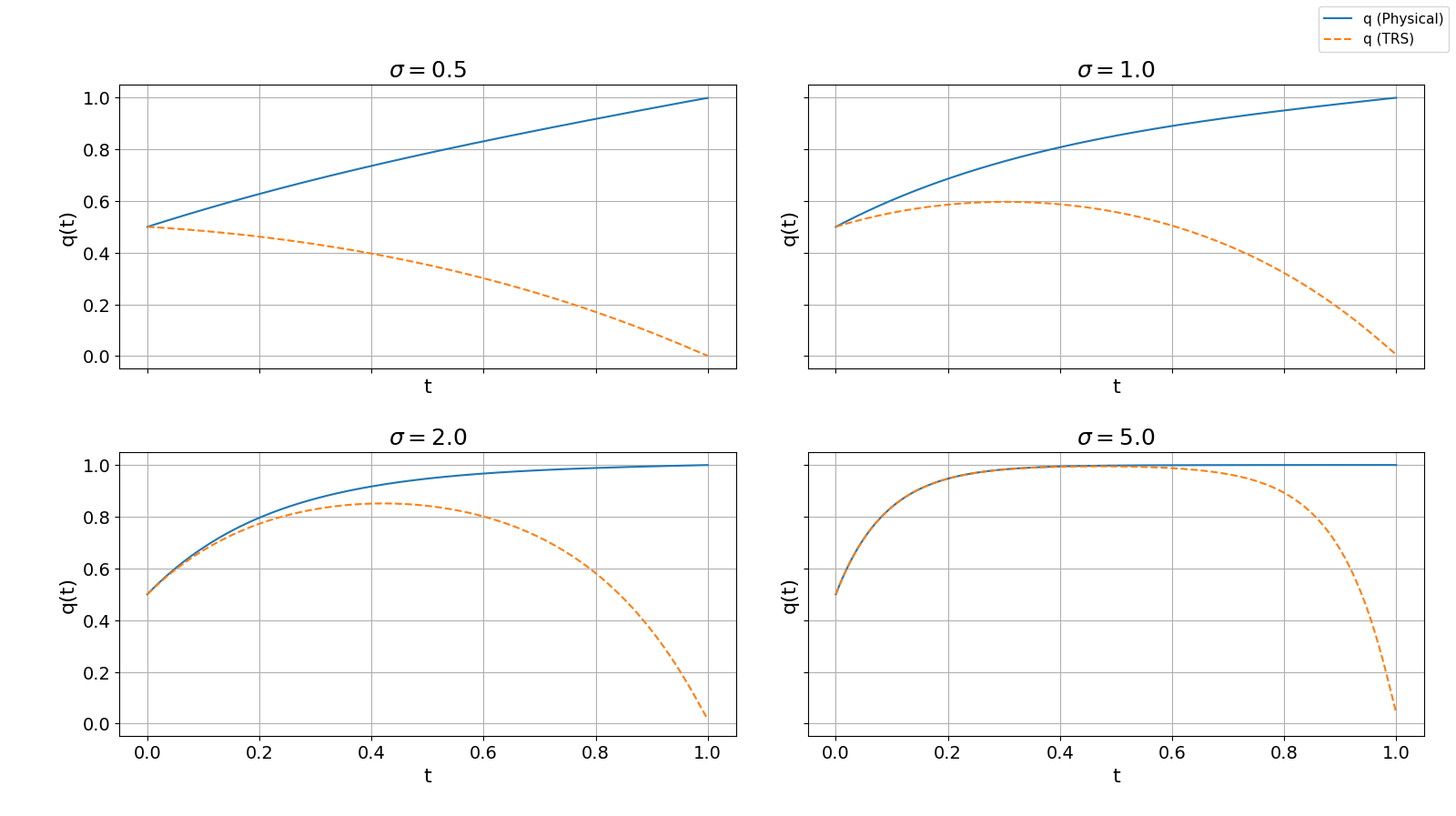}
    \caption{Comparison of the simulations of $Q(t)$ for the physical and TRS contracts under different values of $\sigma$.}
    \label{fig: physical vs TRS under sigma}
\end{figure}

\begin{figure}[h]
    \centering
    \includegraphics[width=0.9\linewidth]{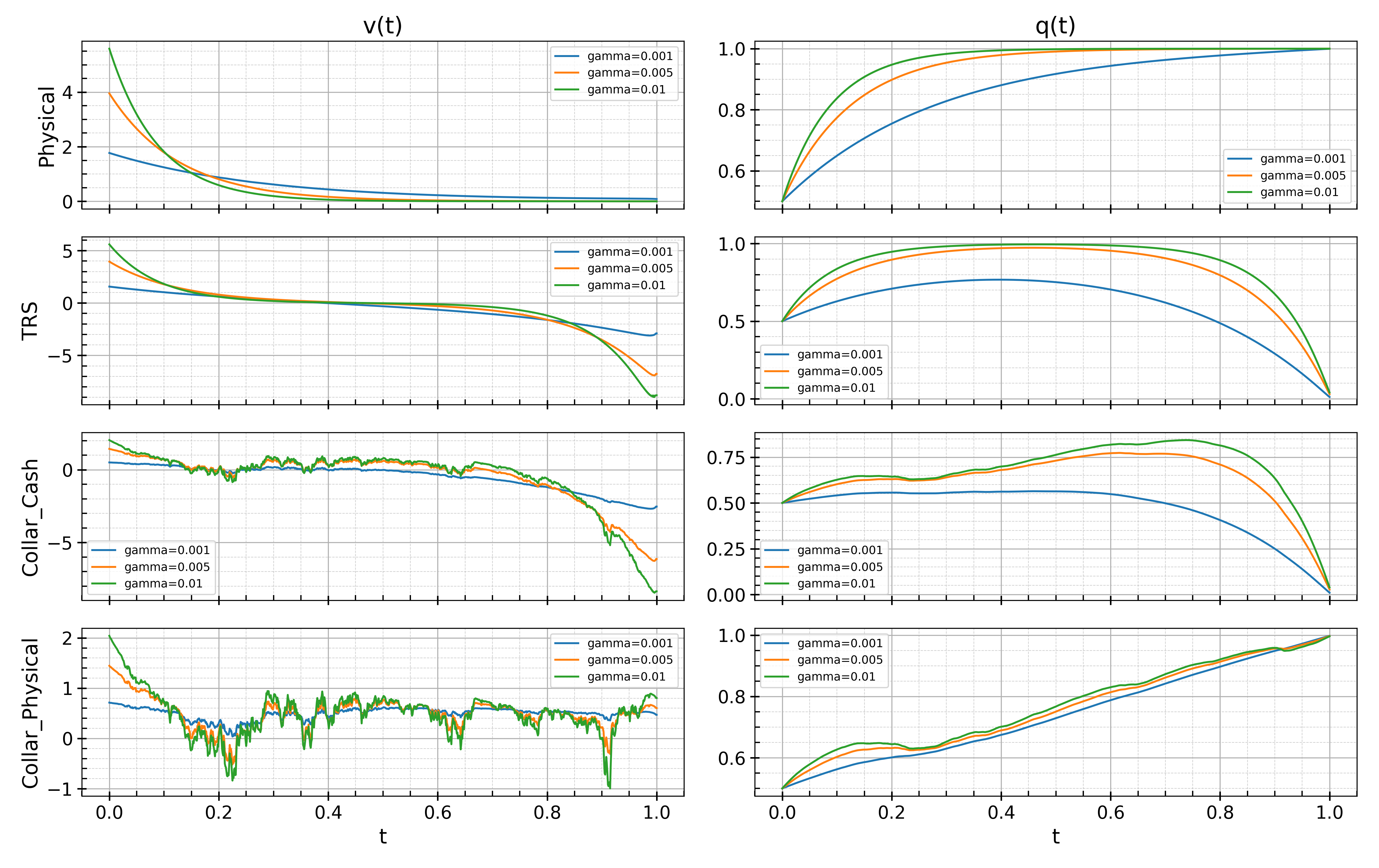}
    \caption{Simulations for different values of $\gamma$.}
    \label{fig:sweep_gamma}
\end{figure}

\begin{figure}[h]
    \centering
    \includegraphics[width=0.8\linewidth]{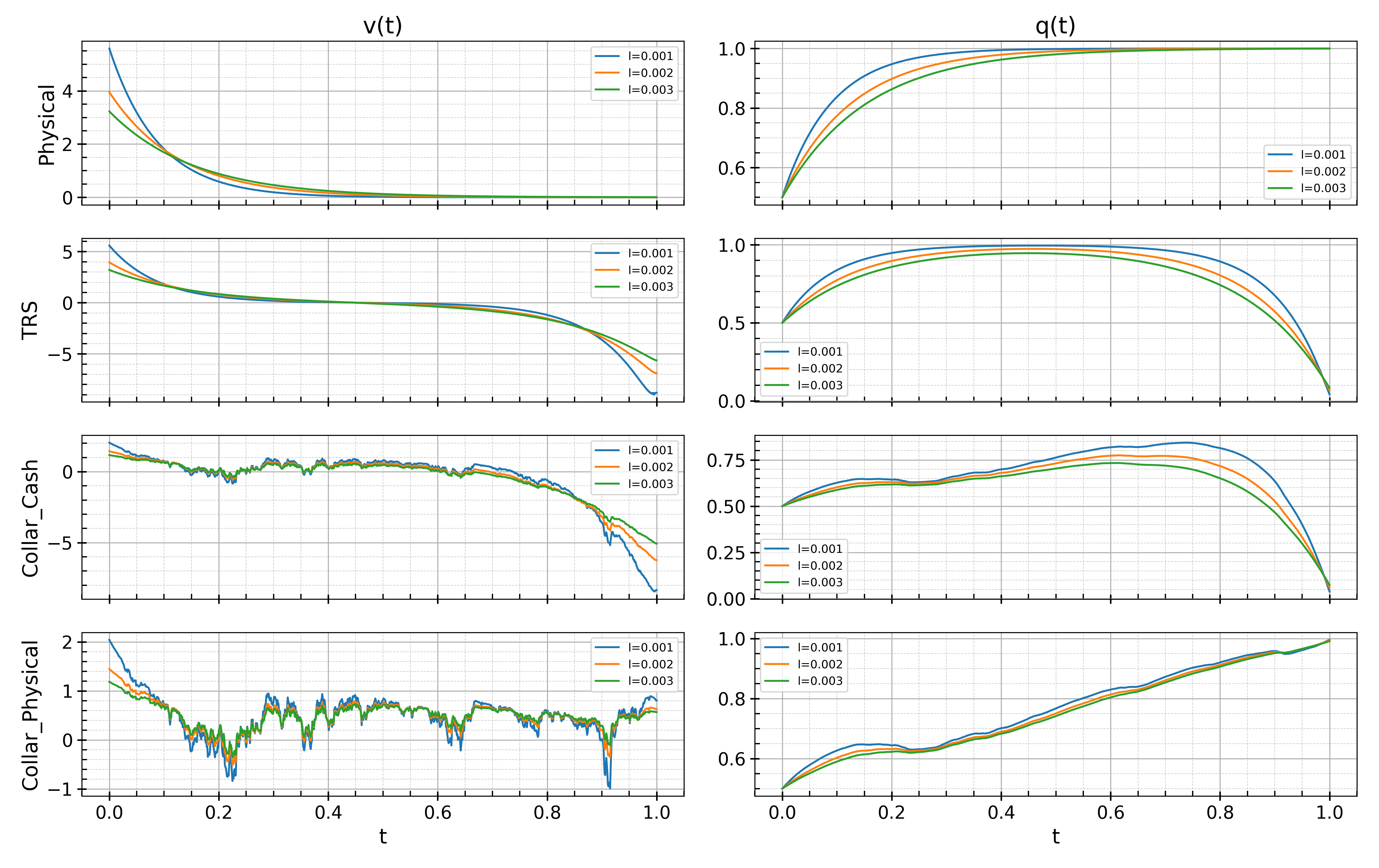}
    \caption{Simulations for different values of $l$.}
    \label{fig:sweep_l}
\end{figure}

\begin{figure}[h]
    \centering
    \includegraphics[width=0.9\linewidth]{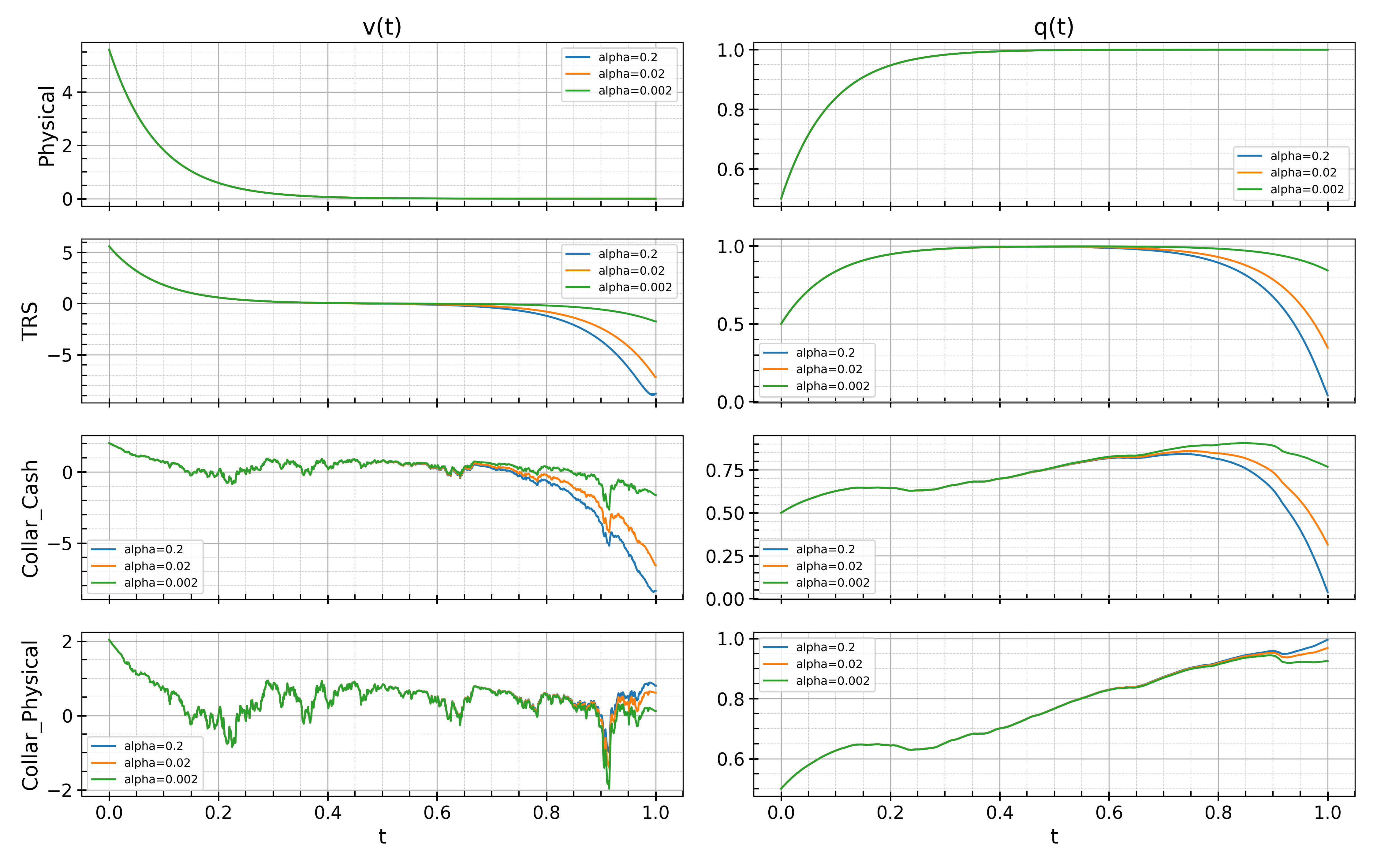}
    \caption{Simulations for different values of $\alpha$.}
    \label{fig:sweep_alpha}
\end{figure}

\begin{figure}[h]
   \centering
  \includegraphics[width=0.9\linewidth]{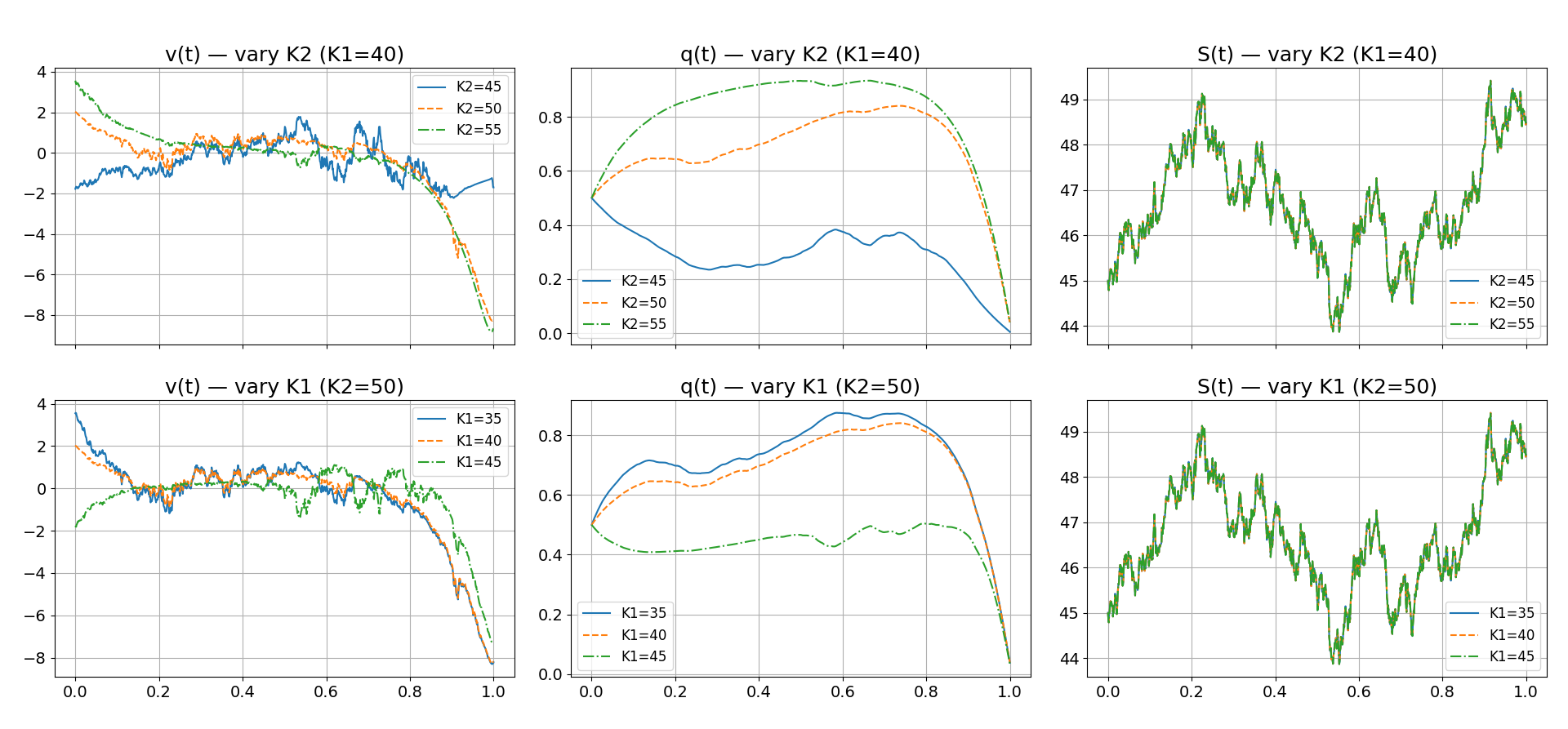}
   \caption{Simulations for different value of $K_1$ and $K_2$ for collar (cash).}
\label{Q2}
\end{figure}

\begin{figure}[h]
   \centering
  \includegraphics[width=0.9\linewidth]{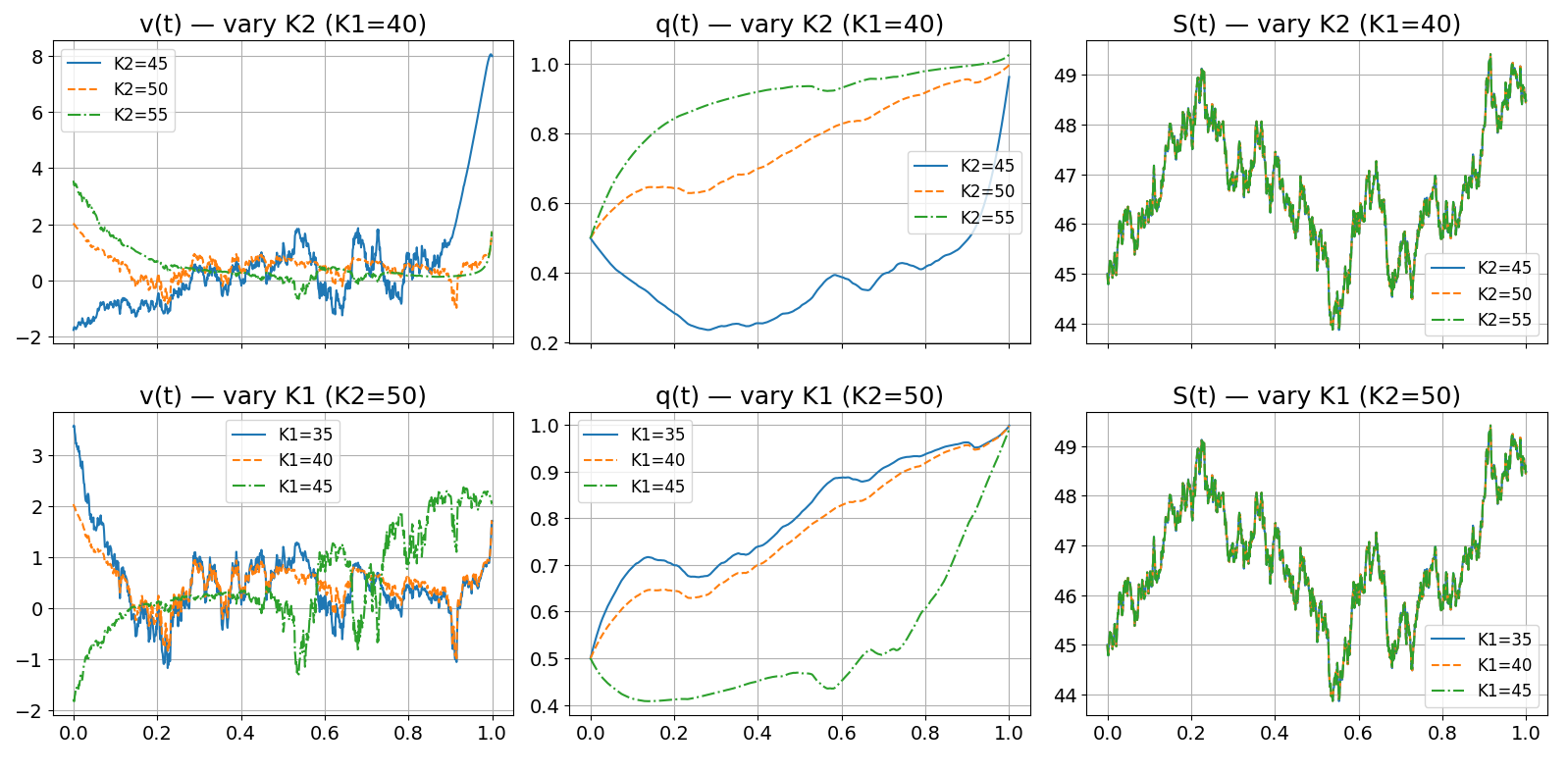}
   \caption{Simulations for different value of $K_1$ and $K_2$ for collar (physical).}
   \label{Q3}
\end{figure}

\end{document}